\documentclass{article}
\usepackage{amsmath}
\usepackage{amssymb}
\usepackage{amsthm}
\usepackage{enumerate}
\usepackage{wasysym}
\usepackage{dsfont}
\usepackage{comment}
\usepackage{tikz}
\usepackage{tikz-cd}

\theoremstyle{definition}

\newtheorem{definition}{Definition}[section]
\newtheorem{theorem}[definition]{Theorem}
\newtheorem{lemma}[definition]{Lemma}

\newtheorem{remark}[definition]{Remark}

\newtheorem{sub-claim}[definition]{Sub-Claim}
\newtheorem{corollary}[definition]{Corollary}

\newcommand{\T}{\mathbb{T}}

\newcommand{\Term}{\mathsf{Term}}
\newcommand{\y}{\mathsf{y}}
\newcommand{\x}{\mathsf{x}}
\newcommand{\mroot}{\mathsf{root}}
\newcommand{\rank}{\mathsf{rank}}
\newcommand{\bmax}{\mathsf{max}}
\newcommand{\Allaliens}{\mathsf{AllAliens}}
\newcommand{\Ind}{\mathsf{Ind}}

\newcommand{\Tmod}{\T\mathsf{mod}}
\newcommand{\cod}{\mathsf{cod}}
\newcommand{\dom}{\mathsf{dom}}
\newcommand{\id}{\mathsf{id}}

\newcommand{\Aut}{\mathsf{Aut}}
\newcommand{\Id}{\mathsf{Id}}

\begin{document}

\title{Isotropy and Combination Problems}
\author{Jason Parker}
\date{\today}

\maketitle

\begin{abstract}
In \cite{MFPSpaper}, the author and his collaborators studied the phenomenon of \emph{isotropy} as introduced in \cite{Funk} in the context of \emph{single-sorted equational theories}, and showed that the isotropy group of the category of models of any such theory encodes a notion of \emph{inner automorphism} for the theory.

Using results from the treatment of \emph{combination problems} in term rewriting theory (cf. \cite[Chapter 9]{Rewriting} and \cite{Combining}), we show in this article that if $\T_1$ and $\T_2$ are (disjoint) equational theories satisfying minimal assumptions, then any free, finitely generated model of the disjoint union theory $\T_1 + \T_2$ has \emph{trivial} isotropy group, and hence the only \emph{inner automorphisms} of such models, i.e. the only automorphisms of such models that are \emph{coherently extendible}, are the \emph{identity} automorphisms. 

As a corollary, we show that the \emph{global isotropy group} of the category of models $(\T_1 + \T_2)\mathsf{mod}$, i.e. the group of invertible elements of the centre of this category, is the trivial group.
\end{abstract}

\section{Background}

In this section we briefly review the relevant background material from \cite{MFPSpaper} and \cite[Chapter 2]{thesis} on the isotropy groups of free models of equational theories. 

A single-sorted \emph{equational theory} $\T$ is a set of equations between the terms of a first-order signature $\Sigma$ consisting of single-sorted operation symbols. For example, the theories of semigroups, (commutative) monoids, (abelian) groups, and (commutative) rings with unit are all single-sorted equational theories. A (set-based) \emph{model} $M$ of a single-sorted equational theory $\T$ is a set equipped with functions on $M$ interpreting the function symbols of the signature, which satisfies the axioms of $\T$. For example, a group is just a model of the equational theory of groups. One can then form the category $\Tmod$ of (set-based) models of $\T$ and homomorphisms between them (i.e. functions that respect the interpretations of the operation symbols). An \emph{isomorphism} in $\T\mathsf{mod}$ can then be easily shown to be just a bijective homomorphism.  
 
Given $M \in \Tmod$, the \emph{(covariant) isotropy group} of $M$ is the group $\mathcal{Z}_\T(M)$ of all natural automorphisms of the forgetful functor $M / \Tmod \to \Tmod$, where $M / \Tmod$ is the slice category under $M$. More concretely, an element of $\mathcal{Z}_\T(M)$ is a family of automorphisms
\[ \pi = \left(\pi_h : \cod(h) \xrightarrow{\sim} \cod(h)\right)_{\dom(h) = M} \] in $\Tmod$ indexed by morphisms $h \in \Tmod$ with domain $M$ that has the following \emph{naturality} property: if $h : M \to M'$ and $h' : M' \to M''$ are homomorphisms in $\Tmod$, then the following diagram commutes:
\begin{center}
\begin{tikzcd}[ampersand replacement=\&, row sep = huge, column sep = huge]
M' \arrow{r}{\pi_{h}}\arrow{d}[swap]{h'} \& M'\arrow{d}{h'} \\
M''\arrow{r}[swap]{\pi_{h' \circ h}} \& M'' \\
\end{tikzcd}
\end{center}     
We then say that an automorphism $h : M \xrightarrow{\sim} M$ of $M \in \Tmod$ is a \emph{(categorical) inner automorphism} (or is \emph{coherently extendible}) if there is some $\pi \in \mathcal{Z}_\T(M)$ with $h = \pi_{\mathsf{id}_M} : M \xrightarrow{\sim} M$; roughly, $h$ is a categorical inner automorphism if it can be coherently extended along any morphism out of $M$. This terminology is motivated by the fact that the categorical inner automorphisms of \emph{groups} are exactly the ordinary inner automorphisms of groups (defined in terms of conjugation), as shown by Bergman in \cite[Theorem 1]{Bergman}. In keeping with this terminology, we will sometimes also refer to the elements of $\mathcal{Z}_\T(M)$ as \emph{extended inner automorphisms} of $M$.

In \cite{MFPSpaper} and \cite{thesis} the author and his collaborators gave a \emph{logical} characterization of the isotropy group of a model of an equational theory. We will only review the (simpler) characterization for the \emph{free, finitely generated} models, as these are the only models that will concern us in this article. If $\T$ is an equational theory and $n \geq 0$, then a model $M_n \in \Tmod$ is \emph{free on} $n$ \emph{generators} if it contains $n$ distinct elements (the \emph{generators}) $m_1, \ldots, m_n$ and has the following universal property: for any $N \in \Tmod$ and any elements $a_1, \ldots, a_n \in N$, there is a unique homomorphism $h_{a_1, \ldots, a_n} : M_n \to N$ with $h_{a_1, \ldots, a_n}(m_i) = a_i$ for each $1 \leq i \leq n$. For a fixed $n \geq 0$, the free $\T$-models on $n$ generators are all isomorphic, so we may speak of \emph{the} free $\T$-model on $n$ generators (unique up to isomorphism). 

The free $\T$-model on $n$ generators can be given the following explicit description: let $\{\y_1, \ldots, \y_n\}$ be a set of $n$ constants, and let $\Sigma(\y_1, \ldots, \y_n)$ be the single-sorted signature obtained from $\Sigma$ by adding the elements $\y_1, \ldots, \y_n$ as new constant symbols. Let $\T(\y_1, \ldots, \y_n)$ be the equational theory with the same axioms as $\T$, but now regarded as an equational theory over the signature $\Sigma(\y_1, \ldots, \y_n)$. Consider the set $\Term^c(\Sigma(\y_1, \ldots, \y_n))$ of \emph{closed} terms over the signature $\Sigma(\y_1, \ldots, \y_n)$ (i.e. terms over this signature that do \emph{not} contain variables). We define a relation $\sim_{\T, n} \ = \ \sim_\T$ on $\Term^c(\Sigma(\y_1, \ldots, \y_n))$ by setting $s \sim_{\T} t$ iff \[ \T(\y_1, \ldots, \y_n) \vdash s = t \] for any $s, t \in \Term^c(\Sigma(\y_1, \ldots, \y_n))$. Roughly, we have $s \sim_{\T} t$ iff $s$ can be proved equal to $t$ using (only) the axioms of $\T$. Then $\sim_{\T}$ is a $\Sigma$-\emph{congruence} relation on $\Term^c(\Sigma(\y_1, \ldots, \y_n))$, i.e. an equivalence relation that is compatible with the function symbols in $\Sigma$. We can then form the \emph{quotient} $\T$-model \[ \Term^c(\Sigma(\y_1, \ldots, \y_n))/{\sim_{\T}}, \] whose objects are $\sim_{\T}$-congruence classes, which will have the desired universal property, with generators $[\y_1], \ldots, [\y_n]$. So we can take \[ M_n :=\Term^c(\Sigma(\y_1, \ldots, \y_n))/{\sim_{\T}} \] as an explicit construction of the free $\T$-model on $n$ generators.

Now let $G_\T(M_n)$ be the set of all elements
\[ [t] \in \Term^c(\Sigma(\x, \y_1, \ldots, \y_n))/{\sim_{\T}} \] (note the additional constant $\x$) with the following properties:
\begin{itemize}
\item $[t]$ is \emph{invertible}, meaning that there is some $t^{-1} \in \Term^c(\Sigma(\x, \y_1, \ldots, \y_n))$ such that
\[ \T(\x, \y_1, \ldots, \y_n) \vdash t[t^{-1}/\x] = \x = t^{-1}[t/\x]. \]

\item $[t]$ \emph{commutes generically with} the operation symbols of $\Sigma$, meaning that if $f$ is an $m$-ary function symbol of $\Sigma$, then
\[ \T(\x_1, \ldots, \x_m, \y_1, \ldots, \y_n) \vdash t[f(\x_1, \ldots, \x_m)/\x] = f(t[\x_1/\x], \ldots, t[\x_m/\x]). \] 
\end{itemize}  
Then this set $G_\T(M_n)$ can be given the structure of a group (with unit element $[\x]$ and multiplication given by substitution into $\x$), and we then have (cf. \cite[Corollary 2.4.15]{thesis})
\[ \mathcal{Z}_\T(M_n) \cong G_\T(M_n). \] We refer to $G_\T(M_n)$ as the \emph{logical} isotropy group of $M_n$; thus, the (categorical) covariant isotropy group of $M_n$ is isomorphic to its \emph{logical} isotropy group. Thus, the extended inner automorphisms of $M$ can essentially be identified with those (congruence classes of) closed $\Sigma$-terms over the constant $\x$ and the generating constants $\y_1, \ldots, \y_n$ that are invertible and commute generically with the operations of $\Sigma$. 

Given any $[t] \in \Term^c(\Sigma(\x, \y_1, \ldots, \y_n))/{\sim_{\T}}$ and any $\T$-model $N$ with $n$ distinct elements $a_1, \ldots, a_n \in N$, the element $[t]$ induces a function
\[ [t]^{N, a_1, \ldots, a_n} : N \to N; \] roughly, given any $a \in N$, the value $[t]^{N, a_1, \ldots, a_n}(a) \in N$ is the element of $N$ obtained by substituting $a_1, \ldots, a_n$ for $\y_1, \ldots, \y_n$ and $a$ for $\x$ in $t$, and then interpreting/evaluating the result in $N$. The following results then follow from the definition of the isomorphism $\mathcal{Z}_\T(M_n) \cong G_\T(M_n)$, cf. \cite[Corollary 2.2.42]{thesis}. For any homomorphism $h : M_n \to N$ in $\Tmod$, let us write $h_1, \ldots, h_n \in N$ for the images $h([\y_1]), \ldots, h([\y_n]) \in N$ of the generators of $M_n$ under $h$.
\begin{itemize}
\item Given any (not necessarily \emph{natural}) family 
\[ \pi = \left(\pi_h : \cod(h) \to \cod(h)\right)_{\dom(h) = M_n} \] of endomorphisms in $\Tmod$ indexed by morphisms with domain $M_n$, we have $\pi \in \mathcal{Z}_\T(M_n)$ iff there is some (uniquely determined) element $[t] \in G_\T(M_n)$ with the property that
\[ \pi_h = [t]^{N, h_1, \ldots, h_n} : N \to N \] for each homomorphism $h : M_n \to N$ in $\Tmod$ with domain $M_n$ (in particular, every such function $[t]^{N, h_1, \ldots, h_n}$ will be a $\T$-model automorphism).    

\item Given any endomorphism $h : M_n \to M_n$ in $\Tmod$, we have that $h$ is a \emph{categorical inner automorphism} of $M_n$ iff there is some element $[t] \in G_\T(M_n)$ with \[ h = [t]^{M, \id_1, \ldots, \id_n} : M \to M \] (where $\id : M_n \to M_n$ is the identity morphism). 
\end{itemize}

\section{Main Result}

If $\T$ is an equational theory over a signature $\Sigma$, then we say that a function symbol $f \in \Sigma$ of arity $n \geq 1$ is a \emph{projection} (in $\T$) if there is some $1 \leq i \leq n$ such that $\T$ proves the equation $f(y_1, \ldots, y_n) = y_i$ for pairwise distinct variables $y_1, \ldots, y_n$. We also say that $f$ is \emph{constant} (in $\T$) if there is some (not necessarily closed) term $s$ over $\Sigma$ for which $\T$ proves the equation $f(y_1, \ldots, y_n) = s$, where $y_1, \ldots, y_n$ are pairwise distinct variables \emph{none of which} occur in $s$.  

Suppose that $\T_1$ and $\T_2$ are equational theories over respective disjoint signatures $\Sigma_1$ and $\Sigma_2$, such that each theory has at least one function symbol that is neither constant nor a projection (in that theory). More precisely, we suppose that $\Sigma_1$ contains some function symbol $f_1$ of arity $n_1 \geq 1$ that is neither constant nor a projection in $\T_1$, and we suppose that $\Sigma_2$ contains a function symbol $f_2$ of arity $n_2 \geq 1$ that is neither constant nor a projection in $\T_2$. 

Let $\T_1 + \T_2$ be the \emph{(disjoint) union} or \emph{combination} of the theories $\T_1$ and $\T_2$, i.e. the equational theory over the single-sorted signature $\Sigma_1 + \Sigma_2 := \Sigma_1 \cup \Sigma_2$ whose axioms are all those of $\T_1$ and $\T_2$ combined. We will show in this section (cf. Theorem \ref{maintheorem}) that every free, finitely generated model of the disjoint union theory $\T_1 + \T_2$ has \emph{trivial} isotropy group. So let $n \geq 0$ and let $M_n$ be the free model of $\T_1 + \T_2$ on the $n$ generators $\y_1, \ldots, \y_n$. Recall from Section 1 that $M_n$ has the explicit description $\Term^c((\Sigma_1 + \Sigma_2)(\y_1, \ldots, \y_n))/{\sim_{\T_1+\T_2}}$, where $\Term^c((\Sigma_1 + \Sigma_2)(\y_1, \ldots, \y_n))$ is the set of all \emph{closed} terms over the signature $(\Sigma_1 + \Sigma_2)(\y_1, \ldots, \y_n)$ obtained from $\Sigma_1 + \Sigma_2$ by adding $n$ pairwise distinct new constant symbols $\y_1, \ldots, \y_n$, and $\sim_{\T_1+\T_2}$ is the $\Sigma_1 + \Sigma_2$-congruence relation on $\Term^c((\Sigma_1 + \Sigma_2)(\y_1, \ldots, \y_n))$ given by 
\[ s \sim_{\T_1+\T_2} t \ \Longleftrightarrow \ (\T_1 + \T_2)(\y_1, \ldots, \y_n) \vdash s = t \] for any $s, t \in \Term^c((\Sigma_1 + \Sigma_2)(\y_1, \ldots, \y_n))$. 

Recall also from Section 1 that the logical isotropy group $G_{\T_1 + \T_2}(M_n)$ consists of all elements
\[ [t] \in \Term^c((\Sigma_1 + \Sigma_2)(\x, \y_1, \ldots, \y_n))/{\sim_{\T_1 + \T_2}} \] (note the additional constant $\x$) with the following properties:
\begin{itemize}
\item $[t]$ is \emph{invertible}, meaning that there is some $t^{-1} \in \Term^c((\Sigma_1 + \Sigma_2)(\x, \y_1, \ldots, \y_n))$ such that
\[ (\T_1 + \T_2)(\x, \y_1, \ldots, \y_n) \vdash t[t^{-1}/\x] = \x = t^{-1}[t/\x]. \]

\item $[t]$ \emph{commutes generically with} the function symbols of $\Sigma_1 + \Sigma_2$, meaning that if $f$ is an $m$-ary function symbol of $\Sigma_1$ or $\Sigma_2$, then
\[ (\T_1 + \T_2)(\x_1, \ldots, \x_m, \y_1, \ldots, \y_n) \vdash t[f(\x_1, \ldots, \x_m)/\x] = f(t[\x_1/\x], \ldots, t[\x_m/\x]). \] 
\end{itemize}  

\noindent Our primary aim in this section (cf. Theorem \ref{maintheorem}) will be to show that this group $G_{\T_1 + \T_2}(M_n)$ is \emph{trivial}, i.e. that $G_{\T_1 + \T_2}(M_n) = \{ [\x] \}$. Let us now try to give some intuition for why this will be the case. If $[t] \in G_{\T_1 + \T_2}(M_n)$ and $[t] \neq [\x]$, then $t \not\equiv \x$, and (unless $\T_1 + \T_2$ is the trivial theory, i.e. no model of $\T_1 + \T_2$ has more than one element) also $t \not\equiv \y_i$ for each $1 \leq i \leq n$. So the `root' or `outer' function symbol of $t$ will belong to either $\Sigma_1$ or $\Sigma_2$, say $\Sigma_1$. Then $[t]$, being an element of the logical isotropy group of $M_n$, must commute generically with all function symbols of $\Sigma_2$. But since the axioms of $\T_1 + \T_2$ are just those of $\T_1$ and $\T_2$ combined, and the signatures $\Sigma_1$ and $\Sigma_2$ are disjoint, it follows that $\T_1 + \T_2$ will not contain any (new) axioms that force the function symbols of the two signatures to interact in any `non-trivial' way. In particular, since the root function symbol of $t$ is in $\Sigma_1$, then $[t]$ will not commute generically with all function symbols of $\Sigma_2$ (more specifically, it will not commute generically with any function symbol of $\Sigma_2$ that is neither constant nor a projection in $\T_2$), and hence $[t] \notin G_{\T_1 + \T_2}(M_n)$, contrary to assumption. So we must have $[t] = [x]$. 
   
To prove this main result, we now need to introduce some notation, definitions, and results from the treatment of combination problems in term rewriting theory (see \cite[Chapter 9]{Rewriting}, \cite{Combining}). We first need to define a new signature \[ \Sigma_3 := \{ \y_1, \ldots, \y_n \} \] containing just the constants $\y_1, \ldots, \y_n$. We also need to define a fourth signature $\Sigma_4$ of `indeterminates' \[ \Sigma_4 := \{ \x \} \cup \{ \x_i : i \geq 1\}, \] where $\x, \x_i$ ($i \geq 1$) are constant symbols. We assume that these signatures $\Sigma_1$ -- $\Sigma_4$ are pairwise disjoint. Finally, we set \[ \Sigma := \Sigma_1 \cup \Sigma_2 \cup \Sigma_3 \cup \Sigma_4. \] 

\noindent Let $V$ be a countable set of variables (disjoint from $\Sigma$), and for $k \in \{1, 2\}$, let $\Term(\Sigma_k, V)$ be the set of all terms over $\Sigma_k$ that (may) contain variables from $V$. For $k \in \{1, 2\}$, a \emph{context} over the signature $\Sigma_k$ is a term $C \in \Term(\Sigma_k, V)$. If $C$ contains at least one variable from $V$ and is not itself a variable, then we say that $C$ is a \emph{proper} context. By a \emph{closed term} over $\Sigma$, we will mean an element of $\Term^c(\Sigma)$, so that a closed term over $\Sigma$ does not contain any variables. A closed term $s$ over $\Sigma$ will be called \emph{pure} if $s \in \Term^c(\Sigma_i)$ for some $1 \leq i \leq 4$ (i.e. $s$ contains symbols from only one signature); otherwise, $s$ will be called \emph{impure}. If $s$ is a closed term over $\Sigma$, then we define $\mroot(s) \in \Sigma$, the `root' symbol of $s$, as follows. If $s$ is a constant, then $\mroot(s) \equiv s$. Otherwise, there is some function symbol $g \in \Sigma_1 \cup \Sigma_2$ of arity $m \geq 1$ such that $s \equiv g(s_1, \ldots, s_m)$ for some closed terms $s_1, \ldots, s_m$ over $\Sigma$, and we define $\mroot(s) := g$. 

Given an impure closed term $s$ over $\Sigma$, if $\mroot(s) \in \Sigma_k$ for $k \in \{1, 2\}$, then there is a proper $\Sigma_k$-context $C$ containing variables $y_1, \ldots, y_\ell \in V$ (for some $\ell \geq 1$) and there are distinct closed terms $s_1, \ldots, s_\ell$ over $\Sigma$ such that $s \equiv C[s_1/y_1, \ldots, s_\ell/y_\ell]$ and $\mroot(s_j) \notin \Sigma_k$ for each $1 \leq j \leq \ell$. Note that $C$ and the (distinct) closed terms $s_1, \ldots, s_\ell$ are uniquely determined by $s$. The closed terms $s_1, \ldots, s_\ell$ are called the \emph{alien} subterms of $s$. We write $s \equiv C[s_1, \ldots, s_\ell]$ with square brackets to indicate that $s_1, \ldots, s_\ell$ are the \emph{alien} subterms of $s$. If one of the terms $s_1, \ldots, s_\ell$ is \emph{not} alien, i.e. if $\mroot(s_j) \in \Sigma_k$ for some $1 \leq j \leq \ell$, then we just write $s \equiv C(s_1, \ldots, s_\ell)$ with round rather than square brackets. 

Next, we define the \emph{rank} of a closed term over $\Sigma$, which essentially measures how many different `levels' the term has, or how impure the term is:

\begin{definition}
Let $s$ be a closed term over $\Sigma$. 
\begin{itemize}
\item If $s$ is pure, then $\rank(s) := 0$. 

\item Otherwise, there is some $k \in \{1, 2\}$ and some (uniquely determined) proper $\Sigma_k$-context $C(y_1, \ldots, y_\ell)$ and distinct alien subterms $s_1, \ldots, s_\ell$ over $\Sigma$ such that $s \equiv C[s_1, \ldots, s_\ell]$, and we then set \[ \rank(s) := 1 + \bmax\{\rank(s_1), \ldots, \rank(s_\ell)\}. \] \qed
\end{itemize}
\end{definition}  

\noindent Now let $\sim$ be an arbitrary equivalence relation on $\Term^c(\Sigma)$. For any $1 \leq k \leq 4$, let $\Term^c(\Sigma_k, \Term^c(\Sigma)/{\sim})$ be the set of all closed terms over $\Sigma_k$ that may contain elements from $\Term^c(\Sigma)/{\sim}$ as constants. We now define the notion of `abstracting alien subterms' with respect to $\sim$ as follows. 

\begin{definition}
Let $1 \leq k \leq 4$. For any term $s \in \Term^c(\Sigma)$, we define $[s]_{\sim}^k \in \Term^c(\Sigma_k, \Term^c(\Sigma)/{\sim})$ as follows: 

\begin{enumerate}

\item If $\mroot(s) \notin \Sigma_k$, then we set $[s]_\sim^k := [s]_\sim$. 

\item If $\mroot(s) \in \Sigma_k$ and $s$ is pure, then we set $[s]_{\sim}^k := s$. 

If $\mroot(s) \in \Sigma_k$ and $s$ is impure and $s \equiv C[s_1, \ldots, s_\ell]$ for some proper $\Sigma_k$-context $C$ and distinct closed terms $s_1, \ldots, s_\ell$, then we set \[ [s]_{\sim}^k := C\left([s_1]_{\sim}, \ldots, [s_\ell]_{\sim}\right). \] \qed
\end{enumerate}
\end{definition}

\noindent Essentially, one abstracts alien subterms from a closed term $s$ over $\Sigma$ by converting all alien subterms of $s$ to their $\sim$-classes. 

We now require the following series of definitions. For $k \in \{1, 2\}$ and $\sim$ an equivalence relation on $\Term^c(\Sigma)$, a $\Sigma_k$-congruence on the set $\Term^c(\Sigma_k, \Term^c(\Sigma)/{\sim})$ is an equivalence relation $R$ on $\Term^c(\Sigma_k, \Term^c(\Sigma)/{\sim})$ that is compatible with the function symbols of $\Sigma_k$, in the sense that if $f$ is an $m$-ary function symbol of $\Sigma_k$ and $s_1, \ldots, s_m, t_1, \ldots, t_m \in \Term^c(\Sigma_k, \Term^c(\Sigma)/{\sim})$ with $s_i \ R \ t_i$ for each $1 \leq i \leq m$, then $f(s_1, \ldots, s_m) \ R \ f(t_1, \ldots, t_m)$.

\begin{definition}
For any equivalence relation $\sim$ on $\Term^c(\Sigma)$ and any $1 \leq k \leq 4$, we define the following relation $\approx_{k, \sim}$ (on $\Term^c(\Sigma_k, \Term^c(\Sigma)/{\sim})$ if $k \in \{1, 2\}$, and on $\Term^c(\Sigma_k) = \Sigma_k$ if $k \in \{3, 4\}$): 

\begin{enumerate}

\item We define $\approx_{1, \sim}$ to be the smallest $\Sigma_1$-congruence relation on \linebreak $\Term^c(\Sigma_1, \Term^c(\Sigma)/{\sim})$ containing all pairs of the form $(\sigma(u), \sigma(v))$, where $u, v \in \Term(\Sigma_1, V)$ and $u = v$ is an axiom of $\T_1$ and $\sigma : V \to \Term^c(\Sigma_1, \Term^c(\Sigma)/{\sim})$ is a substitution. 

\item Similarly, we define $\approx_{2, \sim}$ to be the smallest $\Sigma_2$-congruence relation on \linebreak $\Term^c(\Sigma_2, \Term^c(\Sigma)/{\sim})$ containing all pairs of the form $(\sigma(u), \sigma(v))$, where $u, v \in \Term(\Sigma_2, V)$ and $u = v$ is an axiom of $\T_2$ and $\sigma : V \to \Term^c(\Sigma_2, \Term^c(\Sigma)/{\sim})$ is a substitution. 

\item We define $\approx_{3, \sim}$ to be just the reflexive relation on $\Term^c(\Sigma_3) = \Sigma_3$, i.e. for any $1 \leq j_1, j_2 \leq n$, we have $\y_{j_1} \approx_{3, \sim} \y_{j_2}$ iff $j_1 = j_2$.

\item Similarly, we define $\approx_{4, \sim}$ to be just the reflexive relation on $\Term^c(\Sigma_4) = \Sigma_4$. \qed 
\end{enumerate}
\end{definition}

\begin{definition}
We define the relation $\cong$ on $\Term^c(\Sigma)$ by recursion on the rank of terms. For any $u, v \in \Term^c(\Sigma)$, 

\begin{center}
$u \cong v$ iff $\exists 1 \leq k \leq 4$ such that $\mroot(u), \mroot(v) \in \Sigma_k$ and $[u]_{\cong}^k \approx_k [v]_{\cong}^k$,
\end{center}
\noindent where we have written $\approx_k$ instead of $\approx_{k, \cong}$. \qed
\end{definition}

\noindent Observe that $\cong$ is an equivalence relation, because each $\approx_k$ is one. Also observe that for any $u, v \in \Term^c(\Sigma_3 \cup \Sigma_4) = \Sigma_3 \cup \Sigma_4$, we have $u \cong v$ iff $u \equiv v$. 

Essentially, two closed terms over $\Sigma$ are isomorphic iff their root symbols belong to the same signature $\Sigma_k$ and when one abstracts their alien subterms (with respect to $\cong$), the resulting terms are congruent modulo $\approx_k \ = \ \approx_{k, \cong}$. 

\begin{definition}
Let $\mathcal{A} := \Term^c(\Sigma)/{\cong}$. We turn $\mathcal{A}$ into a $\Sigma$-algebra as follows. For any $1 \leq k \leq 4$ and any $f \in \Sigma_k$ of arity $i \geq 0$, we define $f^{\mathcal{A}} : \mathcal{A}^i \to \mathcal{A}$ as follows. Let $[u_1]_{\cong}, \ldots, [u_i]_{\cong} \in \mathcal{A} = \Term^c(\Sigma)/{\cong}$ be arbitrary.
\begin{itemize}

\item If $u$ is an alien subterm of $f(u_1, \ldots, u_i)$ such that $[f(u_1, \ldots, u_i)]_{\cong}^k \approx_k [u]_{\cong}$, then \[ f^{\mathcal{A}}([u_1]_{\cong}, \ldots, [u_i]_{\cong}) := [u]_{\cong}. \] We say that $f(u_1, \ldots, u_i)$ \emph{collapses} to the alien subterm $u$.

\item Otherwise, if there is no alien subterm of $f(u_1, \ldots, u_i)$ that satisfies the condition above, i.e. if there is no alien subterm to which $f(u_1, \ldots, u_i)$ collapses, then \[ f^{\mathcal{A}}([u_1]_{\cong}, \ldots, [u_i]_{\cong}) := [f(u_1, \ldots, u_i)]_{\cong}. \] \qed
\end{itemize}
\end{definition}

\noindent It is shown in \cite[Lemma 3.7]{Combining} that each $f^{\mathcal{A}}$ is well-defined.

\begin{definition}
For any closed term $s \in \Term^c(\Sigma)$, we can define the interpretation $s^{\mathcal{A}} \in \mathcal{A} = \Term^c(\Sigma)/{\cong}$ as follows:

\begin{enumerate}

\item If $s$ is a constant of $\Sigma$, then $s^{\mathcal{A}} := [s]_{\cong}$.

\item Otherwise, there is some function symbol $g$ of arity $m \geq 1$ of $\Sigma_1 \cup \Sigma_2$ such that $s \equiv g(s_1, \ldots, s_m)$ for some closed terms $s_1, \ldots, s_m \in \Term^c(\Sigma)$. Then 
\[ s^{\mathcal{A}} = g(s_1, \ldots, s_m)^{\mathcal{A}} := g^{\mathcal{A}}\left(s_1^{\mathcal{A}}, \ldots, s_m^{\mathcal{A}}\right). \] \qed
\end{enumerate}
\end{definition}

\noindent We may now prove the following simple lemma:

\begin{lemma}
\label{pureterminterpretation}
{\em If $s \in \Term^c(\Sigma)$ is pure, then $s^{\mathcal{A}} = [s]_{\cong} \in \mathcal{A}$.
}
\end{lemma}

\begin{proof}
We prove this by induction on the structure of $s$. If $s$ is a constant, then by definition we have $s^{\mathcal{A}} := [s]_{\cong}$. Otherwise, there are some $k \in \{1, 2\}$ and some function symbol $g \in \Sigma_k$ of arity $m \geq 1$ and pure closed terms $s_1, \ldots, s_m \in \Term^c(\Sigma_k)$ such that $s \equiv g(s_1, \ldots, s_m)$. Then we have 
\begin{align*}
s^{\mathcal{A}}	&= g(s_1, \ldots, s_m)^{\mathcal{A}} \\
			&= g^{\mathcal{A}}\left(s_1^{\mathcal{A}}, \ldots, s_m^{\mathcal{A}}\right) \\
			&= g^{\mathcal{A}}\left([s_1]_{\cong}, \ldots, [s_m]_{\cong}\right) \\
			&= [g(s_1, \ldots, s_m)]_{\cong} \\
			&= [s]_{\cong},
\end{align*}
\noindent as desired. The third equality follows by the induction hypothesis, and the fourth equality follows because $g(s_1, \ldots, s_m)$ has no alien subterms (to which it could collapse), being pure. This completes the induction.
\end{proof}

\begin{definition}
Let $k \in \{1, 2\}$. If $C(y_1, \ldots, y_\ell) \in \Term(\Sigma_k, V)$ is a $\Sigma_k$-context with free variables $y_1, \ldots, y_\ell \in V$, then $C$ induces a function $C^{\mathcal{A}} : \mathcal{A}^\ell \to \mathcal{A}$ defined as follows:

\begin{enumerate}

\item If $C$ has no free variables, then $C \in \Term^c(\Sigma_k)$, and so $C^{\mathcal{A}} : \mathcal{A}^\ell \to \mathcal{A}$ is just the constant function on $C^{\mathcal{A}} \in \mathcal{A}$.

\item If $C \equiv y_j$ for some $1 \leq j \leq \ell$, then \[ C^{\mathcal{A}} = \pi_j^{\mathcal{A}} : \mathcal{A}^\ell \to \mathcal{A}, \] the $j^{\text{th}}$ projection function on $\mathcal{A}^\ell$. 

\item If $g \in \Sigma_k$ is a function symbol of arity $r \geq 1$ and $C \equiv g(C_1, \ldots, C_r)$ for some $\Sigma_k$-contexts $C_1, \ldots, C_r$ over the free variables $y_1, \ldots, y_\ell$, then for any $a_1, \ldots, a_\ell \in \mathcal{A}$ we set \[ C^{\mathcal{A}}(a_1, \ldots, a_\ell) := g^{\mathcal{A}}\left(C_1^{\mathcal{A}}(a_1, \ldots, a_\ell), \ldots, C_r^{\mathcal{A}}(a_1, \ldots, a_\ell)\right). \] \qed
\end{enumerate}
\end{definition}

\begin{remark}
It is then easy to prove by induction on the structure of contexts that if $C \in \Term(\Sigma_k, V)$ is a $\Sigma_k$-context with free variables $y_1, \ldots, y_\ell$ and $s_1, \ldots, s_\ell$ are any closed terms over $\Sigma$, then 
\[ C(s_1, \ldots, s_\ell)^{\mathcal{A}} = C^{\mathcal{A}}\left(s_1^{\mathcal{A}}, \ldots, s_\ell^{\mathcal{A}}\right) \in \mathcal{A}. \] \qed
\end{remark}

\begin{remark}
The discussion before \cite[Lemma 3.8]{Combining} also indicates that Definition 5 can be extended to contexts. More precisely, if $C(y_1, \ldots, y_\ell) \in \Term(\Sigma_k, V)$ is a proper $\Sigma_k$-context (for $k \in \{1, 2\}$) with variables $y_1, \ldots, y_\ell \in V$, then for any $[u_1]_{\cong}, \ldots, [u_\ell]_{\cong} \in \mathcal{A} = \Term^c(\Sigma)/{\cong}$, we have the following facts:
\begin{itemize}

\item If $u$ is an alien subterm of $C(u_1, \ldots, u_\ell)$ such that $[C(u_1, \ldots, u_\ell)]_{\cong}^k \approx_k [u]_{\cong}$, then $C^{\mathcal{A}}([u_1]_{\cong}, \ldots, [u_\ell]_{\cong}) = [u]_{\cong}$. We say that $C(u_1, \ldots, u_\ell)$ \emph{collapses} to the alien subterm $u$.  

\item Otherwise, if there is no alien subterm of $C(u_1, \ldots, u_\ell)$ that satisfies the condition above, i.e. if there is no alien subterm to which $C(u_1, \ldots, u_\ell)$ collapses, then $C^{\mathcal{A}}([u_1]_{\cong}, \ldots, [u_\ell]_{\cong}) = [C(u_1, \ldots, u_\ell)]_{\cong}$. \qed
\end{itemize}
\end{remark} 

\noindent We will now require the following definitions and lemmas. Let us write $\sim_{M_n}$ for the $\Sigma_1 + \Sigma_2$-congruence relation $\sim_{\T_1 + \T_2}$ on $\Term^c((\Sigma_1 + \Sigma_2)(\y_1, \ldots, \y_n))$, so that 
\[ s \sim_{M_n} t \ \Longleftrightarrow \ (\T_1 + \T_2)(\y_1, \ldots, \y_n) \vdash s = t \] for any $s, t \in \Term^c((\Sigma_1 + \Sigma_2)(\y_1, \ldots, \y_n))$. Let us also write $\sim_{M_n, \x}$ for the similarly defined congruence relation on $\Term^c((\Sigma_1 + \Sigma_2)(\x, \y_1, \ldots, \y_n))$, so that \[ s \sim_{M_n, 
\x} t \ \Longleftrightarrow \ (\T_1 + \T_2)(\x, \y_1, \ldots, \y_n) \vdash s = t \] for any $s, t \in \Term^c((\Sigma_1 + \Sigma_2)(\x, \y_1, \ldots, \y_n))$. It is a standard fact about equational logic that $\sim_{M_n}$ is the smallest $\Sigma_1 + \Sigma_2$-congruence relation on $\Term^c((\Sigma_1 + \Sigma_2)(\y_1, \ldots, \y_n)) = \Term^c(\Sigma_1 \cup \Sigma_2 \cup \Sigma_3)$ that contains all pairs of the form $(\sigma(u), \sigma(v))$, where $u, v \in \Term(\Sigma_1 \cup \Sigma_2, V)$ are such that $u = v$ is an axiom of $\T_1 + \T_2$ and $\sigma : V \to \Term^c(\Sigma_1 \cup \Sigma_2 \cup \Sigma_3)$ is a substitution (i.e. a function), and similarly for $\sim_{M_n, \x}$.     

We now extend the definition of $\sim_{M_n, \x}$ to include other indeterminates from $\Sigma_4$ as well. 

\begin{definition}
For the purposes of this definition, we write $\x \equiv \x_0$, so that $\Sigma_4 = \{\x_i : i \geq 0\}$. Let $J$ be an arbitrary subset of $\mathbb{N} = \{0, 1, 2, \ldots \}$. Then we define $\sim_{M_n, J}$ to be the smallest $\Sigma_1 + \Sigma_2$-congruence relation on $\Term^c(\Sigma_1 \cup \Sigma_2 \cup \Sigma_3 \cup \{\x_j : j \in J\})$ containing all pairs of the form $(\sigma(u), \sigma(v))$, where $u, v \in \Term(\Sigma_1 \cup \Sigma_2, V)$ are such that $u = v$ is an axiom of $\T_1 + \T_2$ and $\sigma : V \to \Term^c(\Sigma_1 \cup \Sigma_2 \cup \Sigma_3 \cup \{\x_j : j \in J\})$ is a substitution.

Notice that $\sim_{M_n, \x} \: = \: \sim_{M_n, \{0\}}$. Also, we write $\sim_{M_n, \infty}$ for $\sim_{M_n, \mathbb{N}}$. \qed
\end{definition}

\noindent Note that two closed terms $s, t \in \Term^c(\Sigma)$ are provably equal in the equational theory $\T_1 + \T_2$ over the signature $\Sigma$ (rather than just over the signature $\Sigma_1 + \Sigma_2$) exactly when $s \sim_{M_n, \infty} t$.

\begin{lemma}
\label{infinitecongruencelemma}
\
\begin{itemize}
{\em
\item Let $I, J \subseteq \mathbb{N}$. If $I \subseteq J$, then $\sim_{M_n, I} \: \subseteq \: \sim_{M_n, J}$. 

\item Let $s, t \in \Term^c(\Sigma)$, and let $J$ be a non-empty subset of $\mathbb{N}$ with the property that if $\x_j$ occurs in $s$ or $t$, then $j \in J$. Then \[ s \sim_{M_n, \infty} t \ \Longrightarrow \ s \sim_{M_n, J} t. \]
}
\end{itemize}

\end{lemma}

\begin{proof}
The proof of the first fact is immediate, because $\sim_{M_n, I}$ is defined to be the smallest congruence relation satisfying certain conditions, and if $I \subseteq J$, then $\sim_{M_n, J}$ satisfies those conditions. 

The second fact follows by \cite[Lemma 5.1.30]{thesis}. 
\end{proof}

\begin{definition}
For any closed term $u$ over $\Sigma$, we define the set $\Allaliens(u)$ of all alien subterms of $u$ of arbitrary depth. We define $\Allaliens(u)$ by recursion on $\rank(u)$ as follows: 

\begin{itemize}

\item If $\rank(u) = 0$, then $\Allaliens(u) := \emptyset$. 

\item If $\rank(u) > 0$, then $u \equiv C[u_1, \ldots, u_\ell]$ for some proper $\Sigma_k$-context $C$ (for some $k \in \{1, 2\}$) and alien subterms $u_1, \ldots, u_\ell$ over $\Sigma$. Then we set 
\[ \Allaliens(u) := \{u_1, \ldots, u_\ell\} \cup \bigcup_{1 \leq j \leq \ell} \Allaliens(u_j). \] \qed
\end{itemize}
\end{definition}

\noindent We will require the following lemma, which essentially states that any closed term over $\Sigma_1 \cup \Sigma_2 \cup \Sigma_3 \cup \{\x\}$ is congruent (modulo $\sim_{M_n, \x}$) to a term of minimal rank whose interpretation in $\mathcal{A}$ is just its $\cong$-class: 

\begin{lemma}
\label{canonicaltermlemma}
{\em
For any closed term $t$ over $\Sigma_1 \cup \Sigma_2 \cup \Sigma_3 \cup \{\x\}$, there is some closed term $t'$ over $\Sigma_1 \cup \Sigma_2 \cup \Sigma_3 \cup \{\x\}$ with the following properties:

\begin{enumerate}

\item $\rank(t') \leq \rank(t)$.

\item $t \sim_{M_n, \x} t'$.

\item $(t')^{\mathcal{A}} = [t']_{\cong}$, and if $u \in \Allaliens(t')$, then $u^{\mathcal{A}} = [u]_{\cong}$.

\item $t'$ is not isomorphic to any term over $\Sigma_1 \cup \Sigma_2 \cup \Sigma_3 \cup \{\x\}$ of strictly smaller rank.
\end{enumerate}
}
\end{lemma}

\begin{proof}
We prove the result by induction on $\rank(t)$. For the base case, suppose that $\rank(t) = 0$, so that $t$ is pure. We show that we can take $t' := t$. Conditions 1 and 2 are clearly satisfied. By Lemma \ref{pureterminterpretation}, since $t$ is pure, we have $t^{\mathcal{A}} = [t]_{\cong}$, satisfying the first part of Condition 3. Since $t$ is pure, $t$ has no alien subterms, so the second part of Condition 3 is vacuously satisfied. Since $\rank(t) = 0$, Condition 4 is obviously satisfied. \par

Now suppose that $\rank(t) = m + 1$ for some $m \geq 0$. Then there are some $k \in \{1, 2\}$ and some proper $\Sigma_k$-context $C(y_1, \ldots, y_\ell) \in \Term(\Sigma_k, V)$ such that $t \equiv C[s_1, \ldots, s_\ell]$ for some closed terms $s_1, \ldots, s_\ell$ over $\Sigma_1 \cup \Sigma_2 \cup \Sigma_3 \cup \{\x\}$ with $\mroot(s_j) \notin \Sigma_k$ for each $1 \leq j \leq \ell$. Since $\rank(t) = m + 1$, it follows that $\rank(s_j) \leq m$ for each $1 \leq j \leq \ell$, so by the induction hypothesis, for each $1 \leq j \leq \ell$ there is a closed term ${s_j}'$ over $\Sigma_1 \cup \Sigma_2 \cup \Sigma_3 \cup \{\x\}$ that satisfies Conditions 1 through 4 (with $s_j$ in place of $t$). \par

Now consider the closed term $C\left({s_1}', \ldots, {s_\ell}'\right)$. Then we have $C(s_1, \ldots, s_\ell) \sim_{M_n, \x} C({s_1}', \ldots, {s_\ell}')$, since $\sim_{M_n, \x}$ is a $\Sigma_1 + \Sigma_2$-congruence and $s_j \sim_{M_n, \x} {s_j}'$ for each $1 \leq j \leq \ell$ by Condition 2. Also, we have 
\[ C({s_1}', \ldots, {s_\ell}')^{\mathcal{A}} = C^{\mathcal{A}}\left(({s_1}')^{\mathcal{A}}, \ldots, ({s_\ell}')^{\mathcal{A}}\right) = C^{\mathcal{A}}\left([{s_1}']_{\cong}, \ldots, [{s_\ell}']_{\cong}\right), \]
\noindent with the first equality justified by Remark 1 and the second by the first part of Condition 3 for each $s_j$. Now we must compute $C^{\mathcal{A}}([{s_1}']_{\cong}, \ldots, [{s_\ell}']_{\cong})$, by considering whether $C({s_1}', \ldots, {s_\ell}')$ collapses to any alien subterm (cf. Remark 2). 
\begin{itemize}
\item Suppose first that $C({s_1}', \ldots, {s_\ell}')$ does \emph{not} collapse to any alien subterm. Then we have 
\[ C({s_1}', \ldots, {s_\ell}')^{\mathcal{A}} = C^{\mathcal{A}}([{s_1}']_{\cong}, \ldots, [{s_\ell}']_{\cong}) = [C({s_1}', \ldots, {s_\ell}')]_{\cong}. \]

\noindent Now we show that $C({s_1}', \ldots, {s_\ell}')$ satisfies Conditions 1 to 3. To show that Condition 1 is satisfied, we must show that $\rank(C({s_1}', \ldots, {s_\ell}')) \leq \rank(C[s_1, \ldots, s_\ell])$. If ${s_j}'$ is a pure $\Sigma_k$-term for each $1 \leq j \leq \ell$, then $\rank(C({s_1}', \ldots, {s_\ell}')) = 0$ and the claim holds. Otherwise, there is at least one $1 \leq j \leq \ell$ such that ${s_j}'$ is \emph{not} a pure $\Sigma_k$-term, so write $C({s_1}', \ldots, {s_\ell}') \equiv C'[u_1, \ldots, u_r]$ for some proper $\Sigma_k$-context $C'$ and closed terms $u_1, \ldots, u_r$ over $\Sigma_1 \cup \Sigma_2 \cup \Sigma_3 \cup \{\x\}$ such that $\mroot(u_q) \notin \Sigma_k$ for every $1 \leq q \leq r$. Furthermore, for every $1 \leq q \leq r$, there is some $1 \leq j \leq \ell$ such that either $u_q \equiv {s_j}'$ (if $\mroot({s_j}') \notin \Sigma_k$), or $u_q$ is an alien subterm of ${s_j}'$ (if $\mroot({s_j}') \in \Sigma_k$). So for every $1 \leq q \leq r$, there is some $1 \leq j \leq \ell$ such that $\rank(u_q) \leq \rank({s_j}')$, and hence $\bmax\{\rank(u_1), \ldots, \rank(u_r)\} \leq \bmax\{\rank({s_1}'), \ldots, \rank({s_\ell}')\}$. Thus, we have
\begin{align*}
\rank(C({s_1}', \ldots, {s_\ell}'))	&= \rank(C'[u_1, \ldots, u_r]) \\
						&= 1 + \bmax\{\rank(u_1), \ldots, \rank(u_r)\} \\
						&\leq 1 + \bmax\{\rank({s_1}'), \ldots, \rank({s_\ell}')\} \\
						&\leq 1 + \bmax\{\rank({s_1}), \ldots, \rank({s_\ell})\} \\
						&= \rank(C[s_1, \ldots, s_\ell]) \\
						&= \rank(t),
\end{align*}
\noindent as desired, with the fourth inequality justified by the induction hypothesis. We already noted above that Condition 2 is satisfied for $C({s_1}', \ldots, {s_l}')$, and the first part of Condition 3 was shown above, after we supposed that $C({s_1}', \ldots, {s_l}')$ does not collapse to any alien subterm. Now we must show that the second part of Condition 3 holds for $C({s_1}', \ldots, {s_l}')$. If $C({s_1}', \ldots, {s_l}')$ happens to have no alien subterms, then this part of Condition 3 vacuously holds. Otherwise, as above, let $C({s_1}', \ldots, {s_l}') \equiv C'[u_1, \ldots, u_r]$, so that $u_1, \ldots, u_r$ are the alien subterms of $C({s_1}', \ldots, {s_l}')$. For every $1 \leq q \leq r$, we must show that $u_q^{\mathcal{A}} = [u_q]_{\cong}$, and that $v^{\mathcal{A}} = [v]_{\cong}$ for every $v \in \Allaliens\left(u_q\right)$. If $1 \leq q \leq r$, then there is some $1 \leq j \leq \ell$ such that either $u_q \equiv {s_j}'$, or $u_q$ is an alien subterm of ${s_j}'$. In the first case, the desired result follows by Condition 3 for ${s_j}'$. Now suppose that the second case holds. Since Condition 3 holds for ${s_j}'$ and $u_q \in \Allaliens({s_j}')$, it follows that $u_q^{\mathcal{A}} = [u_q]_{\cong}$. Also, we have $\Allaliens(u_q) \subseteq \Allaliens({s_j}')$, so the second part of Condition 3 holds for $u_q$ because it holds for ${s_j}'$. So both parts of Condition 3 hold for $C({s_1}', \ldots, {s_\ell}')$. \par

Now, if $C({s_1}', \ldots, {s_\ell}')$ already satisfies Condition 4, then we may take $t' := C({s_1}', \ldots, {s_\ell}')$. Otherwise, let $u$ be a term of minimal rank over $\Sigma_1 \cup \Sigma_2 \cup \Sigma_3 \cup \{\x\}$ such that $\rank(u) < \rank(C({s_1}', \ldots, {s_\ell}'))$ and $C({s_1}', \ldots, {s_\ell}') \cong u$. Since $\rank(u) < \rank(C({s_1}', \ldots, {s_\ell}')) \leq \rank(t)$ (as shown above), we can appeal to the induction hypothesis to obtain a closed term $u'$ over $\Sigma_1 \cup \Sigma_2 \cup \Sigma_3 \cup \{\x\}$ that satisfies Conditions 1 to 4 (with $u$ in place of $t$). Now we show that we can take $t' := u'$. First, Condition 1 is satisfied because we have $\rank(u') \leq \rank(u) < \rank(t)$. Since $C({s_1}', \ldots, {s_\ell}') \cong u$, it follows by \cite[Lemma 3.5]{Combining} that $C({s_1}', \ldots, {s_\ell}') \sim_{M_n, \infty} u$, which implies $C({s_1}', \ldots, {s_\ell}') \sim_{M_n, \x} u$ by Lemma \ref{infinitecongruencelemma} (because $C({s_1}', \ldots, {s_\ell}')$ and $u$ are closed terms over $\Sigma_1 \cup \Sigma_2 \cup \Sigma_3 \cup \{\x\}$). So Condition 2 is satisfied, since we have $t \sim_{M_n, \x} C({s_1}', \ldots, {s_\ell}') \sim_{M_n, \x} u \sim_{M_n, \x} u'$. Conditions 3 and 4 are true for $u'$, by definition of $u'$. So we may take $t' := u'$. 

\item Now suppose that $C({s_1}', \ldots, {s_\ell}')$ \emph{does} collapse to an alien subterm $u$. Then there is some $1 \leq j \leq \ell$ such that either $u \equiv {s_j}'$ (if $\mroot({s_j}') \notin \Sigma_k$) or $u$ is an alien subterm of ${s_j}'$ (if $\mroot({s_j}') \in \Sigma_k$). And we have \[ C({s_1}', \ldots, {s_\ell}')^\mathcal{A} = [u]_{\cong}. \] Suppose first that $u \equiv {s_j}'$. Then we show that $t' := u \equiv {s_j}'$ satisfies Conditions 1 to 4. First, since ${s_j}'$ is a closed term over $\Sigma_1 \cup \Sigma_2 \cup \Sigma_3 \cup \{\x\}$, so is $t'$. Then, we have
\[ \rank(t') = \rank(u) = \rank({s_j}') \leq \rank(s_j) < \rank(t), \] so that Condition 1 is satisfied. To show that $t \sim_{M_n, \x} u$, note that since $C({s_1}', \ldots, {s_\ell}')^\mathcal{A} = [u]_{\cong}$, it follows by \cite[Lemma 3.8]{Combining} that $C({s_1}', \ldots, {s_\ell}') \sim_{M_n, \infty} u$, and hence that $C({s_1}', \ldots, {s_\ell}') \sim_{M_n, \x} u$ by Lemma \ref{infinitecongruencelemma} (because $C({s_1}', \ldots, {s_\ell}')$ and $u$ are closed terms over $\Sigma_1 \cup \Sigma_2 \cup \Sigma_3 \cup \{\x\}$). So then we have $t \sim_{M_n, \x} C({s_1}', \ldots, {s_\ell}') \sim_{M_n, \x} u$, as desired. Finally, $t' \equiv {s_j}'$ satisfies Conditions 3 and 4 because ${s_j}'$ does, by the induction hypothesis. \par

Now suppose that $u$ is an alien subterm of ${s_j}'$. Then $\rank(u) < \rank({s_j}') \leq \rank(s_j) < \rank(t)$. So by the induction hypothesis, there is a closed term $u'$ over $\Sigma_1 \cup \Sigma_2 \cup \Sigma_3 \cup \{\x\}$ that satisfies Conditions 1 to 4 (with $u$ in place of $t$). Now we show that $t' := u'$ satisfies Conditions 1 to 4. First, Condition 1 is satisfied, because we have
\[ \rank(t') = \rank(u') \leq \rank(u) < \rank({s_j}') \leq \rank(s_j) < \rank(t). \] Next, Condition 2 is satisfied, because we have \[ t \sim_{M_n, \x} C({s_1}', \ldots, {s_\ell}') \sim_{M_n, \x} u \sim_{M_n, \x} u'. \] And we know that $u'$ satisfies Conditions 3 and 4, by definition of $u'$. So $u'$ satisfies Conditions 1 to 4. \par
\end{itemize}
This completes the induction on $\rank(t)$, and hence completes the proof of the lemma.
\end{proof}

\noindent We now require the following technical lemma, which essentially states that the relation $\approx_k$ (for $k \in \{1, 2\}$) is `closed' under certain kinds of substitutions:

\begin{lemma}
\label{closedundersubs}
{\em
Let $k \in \{1, 2\}$. Then $\approx_k$ is closed under well-defined substitutions on $\Term^c(\Sigma)/{\cong}$ in the following two senses. First, if $\sigma : \Term^c(\Sigma)/{\cong} \to \Term(\Sigma_k, V)$ is a well-defined function, then for any $u, v \in \Term^c(\Sigma_k, \Term^c(\Sigma)/{\cong})$, we have \[ u \approx_k v \Longrightarrow \T_k \vdash \sigma^*(u) = \sigma^*(v), \]

\noindent where $\sigma^*$ is the extension of $\sigma$ from $\Term^c(\Sigma)/{\cong}$ to $\Term^c(\Sigma_k, \Term^c(\Sigma)/{\cong})$. 

Moreover, if $\sigma : \Term^c(\Sigma)/{\cong} \to \Term^c(\Sigma_k, \Term^c(\Sigma)/{\cong})$ is a well-defined function, then for any $u, v \in \Term^c(\Sigma_k, \Term^c(\Sigma)/{\cong})$, we have \[ u \approx_k v \Longrightarrow \sigma^*(u) \approx_k \sigma^*(v). \]
}
\end{lemma}

\begin{proof}
By definition, since $k \in \{1, 2\}$, we know that $\approx_k$ is the smallest $\Sigma_k$-congruence relation on $\Term^c(\Sigma_k, \Term^c(\Sigma)/{\cong})$ containing all pairs of the form $(\tau(w), \tau(w'))$, where $w, w' \in \Term(\Sigma_k, V)$ and $w = w'$ is an axiom of $\T_k$ and $\tau : V \to \Term^c(\Sigma_k, \Term^c(\Sigma)/{\cong})$ is a substitution. To show that $\approx_k$ is closed under well-defined substitutions in the first sense, let $\sigma : \Term^c(\Sigma)/{\cong} \to \Term(\Sigma_k, V)$ be a well-defined function. \par
First let $w, w' \in \Term(\Sigma_k, V)$ be such that $w = w'$ is an axiom of $\T_k$, and let $\tau : V \to \Term^c(\Sigma_k, \Term^c(\Sigma)/{\cong})$ be a substitution. We must show that $\T_k \vdash \sigma^*(\tau(w)) = \sigma^*(\tau(w'))$. But $\sigma^* \circ \tau : V \to \Term(\Sigma_k, V)$ is a well-defined substitution, so because $w = w'$ is an axiom of $\T_k$, it follows that $\T_k \vdash (\sigma^* \circ \tau)(w) = (\sigma^* \circ \tau)(w')$ and hence $\T_k \vdash \sigma^*(\tau(w)) = \sigma^*(\tau(w'))$, as desired. \par
It remains to show that \[ \left\{(u, v) \in \Term^c(\Sigma_k, \Term^c(\Sigma)/{\cong})^2 : \T_k \vdash \sigma^*(u) = \sigma^*(v) \right\} \] is a $\Sigma_k$-congruence relation on $\Term^c(\Sigma_k, \Term^c(\Sigma)/{\cong})$. It is obviously an equivalence relation. To show that it is a $\Sigma_k$-congruence, let $g$ be a function symbol of $\Sigma_k$ of arity $m \geq 1$, and let $u_1, \ldots, u_m, v_1, \ldots, v_m \in \Term^c(\Sigma_k, \Term^c(\Sigma)/{\cong})$ with $\T_k \vdash \sigma^*(u_i) = \sigma^*(v_i)$ for each $1 \leq i \leq m$. We must show that $\T_k \vdash \sigma^*(g(u_1, \ldots, u_m)) = \sigma^*(g(v_1, \ldots, v_m))$. But this follows by the assumption, because provable equality in $\T_k$ is a $\Sigma_k$-congruence, and because $\sigma^*(g(u_1, \ldots, u_m)) \equiv g(\sigma^*(u_1), \ldots, \sigma^*(u_m))$ and $\sigma^*(g(v_1, \ldots, v_m)) \equiv g(\sigma^*(v_1), \ldots, \sigma^*(v_m))$. This completes the argument. The proof that $\approx_k$ is also closed under well-defined substitutions in the second sense is extremely similar.  
\end{proof}

\noindent The following lemma roughly states that if two closed terms with root symbols in the same signature $\Sigma_k$ ($k \in \{1, 2\}$) are congruent modulo $\approx_k$ after abstracting alien subterms, then we can substitute (distinct) \emph{variables} for these alien subterms, and the resulting terms will then be provably equal in $\T_k$.

\begin{lemma}
\label{complexlemma}
{\em
Let $k \in \{1, 2\}$, and let $C(y_1, \ldots, y_\ell), D(z_1, \ldots, z_r) \in \Term(\Sigma_k, V)$ be $\Sigma_k$-contexts, where $C$ or $D$ could just be a single variable (but $C$ and $D$ must each contain at least one variable). Let $s_1, \ldots, s_\ell, t_1, \ldots, t_r$ be closed terms over $\Sigma$ such that for all $1 \leq j \leq \ell$ and all $1 \leq q \leq r$, we have $\mroot(s_j) \notin \Sigma_k$ and $\mroot(t_q) \notin \Sigma_k$. Suppose that the terms $s_1, \ldots, s_\ell, t_1, \ldots, t_r$ can be partitioned into $p \geq 1$ $\cong$-classes. For any $1 \leq j \leq \ell$, let $1 \leq p_j \leq p$ be such that $s_j$ belongs to the $p_j^{\text{th}}$ $\cong$-class, and for any $1 \leq q \leq r$, let $1 \leq p_{\ell+q} \leq p$ be such that $t_q$ belongs to the $p_{\ell+q}^{\text{th}}$ $\cong$-class. 

Let $w_1, \ldots, w_p$ be new variables. Then
\[ \left[C[s_1, \ldots, s_\ell]\right]_{\cong}^k \approx_k \left[D[t_1, \ldots, t_r]\right]_{\cong}^k \]
\[ \Longrightarrow \T_k \vdash C(w_{p_1}/y_1, \ldots, w_{p_\ell}/y_\ell) = D(w_{p_{\ell+1}}/z_1, \ldots, w_{p_{\ell+r}}/z_r). \]
}
\end{lemma}

\begin{proof}
Assume that $[C[s_1, \ldots, s_\ell]]_{\cong}^k \approx_k [D[t_1, \ldots, t_r]]_{\cong}^k$. Then we have 
\[ C([s_1]_{\cong}, \ldots, [s_\ell]_{\cong}) \approx_k D([t_1]_{\cong}, \ldots, [t_r]_{\cong}), \] since $C$ and $D$ are $\Sigma_k$-contexts. Now define the substitution $\sigma : \Term^c(\Sigma)/{\cong} \to \Term(\Sigma_k, V)$ as follows. Let $u \in \Term^c(\Sigma)$ be arbitrary. If $u$ is isomorphic to some term $u'$ in the list $s_1, \ldots, s_\ell, t_1, \ldots, t_r$, then say that $u'$ belongs to the $p_0^{\text{th}}$ $\cong$-class, for some unique $1 \leq p_0 \leq p$, and set $\sigma([u]_{\cong}) := w_{p_0}$. Otherwise, set $\sigma([u]_{\cong}) := z$ for some fixed variable $z \in V$ distinct from $w_1, \ldots, w_p$. \par

Now we show that $\sigma$ is well-defined. So let $u, v \in \Term^c(\Sigma)$ with $u \cong v$. We must show that $\sigma([u]_{\cong}) = \sigma([v]_{\cong})$. If neither $u$ nor $v$ are isomorphic to any term in the list $s_1, \ldots, s_\ell, t_1, \ldots, t_r$, then we have $\sigma([u]_{\cong}) = z = \sigma([v]_{\cong})$, as desired. Otherwise, suppose that at least one of $u$ or $v$ is isomorphic to a term in this list. Since $u \cong v$, it then follows that both terms are isomorphic to terms in the list, which must belong to the same $\cong$-class of terms in the list, say the $p_0^{\text{th}}$ $\cong$-class, for some unique $1 \leq p_0 \leq p$. Then we have $\sigma([u]_{\cong}) = w_{p_0} = \sigma([v]_{\cong})$, as desired. \par

By Lemma \ref{closedundersubs} and the assumption, it then follows that 
\[ \T_k \vdash \sigma^*(C([s_1]_{\cong}, \ldots, [s_\ell]_{\cong})) = \sigma^*(D([t_1]_{\cong}, \ldots, [t_r]_{\cong})), \] which implies that
\[ \T_k \vdash C(\sigma^*([s_1]_{\cong}), \ldots, \sigma^*([s_\ell]_{\cong})) = D(\sigma^*([t_1]_{\cong}), \ldots, \sigma^*([t_r]_{\cong})), \] and hence that 
\[ \T_k \vdash C(w_{p_1}/y_1, \ldots, w_{p_\ell}/y_\ell) = D(w_{p_{\ell+1}}/z_1, \ldots, w_{p_{\ell+r}}/z_r), \] as desired.
\end{proof}

\noindent We will require the following two technical lemmas. 

\begin{lemma}
\label{indeterminatelemma}
{\em
For any $i \geq 1$ and closed terms $u, v$ over $\Sigma \setminus \{\x_i\}$, \[ u[\x_i/\x] \cong v[\x_i/\x] \Longleftrightarrow u \cong v. \] 
}
\end{lemma}

\begin{proof}
Let $u, v \in \Term^c(\Sigma \setminus \{\x_i\})$. We prove by induction on \[ \bmax\{\rank(u), \rank(v)\} \geq 0 \] that the desired claim holds. First suppose that $\rank(u) = \rank(v) = 0$, so that $u$ and $v$ are pure, and assume that $u \cong v$. Now, either $u \equiv \x$ or $u$ does not contain $\x$. If $u \equiv \x$, then since $u \cong v$, it follows that $v \equiv \x$ as well. Then we have $u[\x_i/\x] \equiv \x_i \cong \x_i \equiv v[\x_i/\x]$, as desired. If $u$ does not contain $\x$, then since $u \cong v$, it follows that $v \not\equiv \x$, so that $v$ does not contain $\x$ either (since $v$ is pure). But then we have $u[\x_i/\x] \equiv u \cong v \equiv v[\x_i/\x]$, as desired. \par Conversely, suppose that $u[\x_i/\x] \cong v[\x_i/\x]$. Again, either $u \equiv \x$ or $u$ does not contain $\x$. In the first case, we obtain $u[\x_i/\x] \equiv \x_i \cong v[\x_i/\x]$. It then follows that $v[\x_i/\x] \equiv \x_i$, so that either $v \equiv \x$ or $v \equiv \x_i$. In the first case, since $u \equiv \x$ as well, we obtain $u \cong v$, as desired. And the second case is impossible, because $v \in \Term^c(\Sigma \setminus \{\x_i\})$. Now suppose that $u$ does not contain $\x$. Then we have $u \equiv u[\x_i/\x] \cong v[\x_i/\x]$. Now, either $v \equiv \x$ or $v$ does not contain $\x$. In the first case, we would have $u \cong v[\x_i/\x] \equiv \x_i$, which is impossible, since $u \not\equiv \x_i$, since $u \in \Term^c(\Sigma \setminus \{\x_i\})$. In the second case, we have $u \cong v[\x_i/\x] \equiv v$, as desired. This completes the proof of the base case. \par

Now suppose that $\bmax\{\rank(u), \rank(v)\} = m > 0$, and assume that the result holds for all $u', v' \in \Term^c(\Sigma \setminus \{\x_i\})$ with $\bmax\{\rank(u'), \rank(v'))\} < m$. Since $\bmax\{\rank(u), \rank(v)\} > 0$, either $\rank(u) > 0$ or $\rank(v) > 0$. Suppose without loss of generality that $\rank(u) > 0$, and let $u \equiv C[s_1, \ldots, s_\ell]$ for some $\Sigma_k$-context $C \in \Term(\Sigma_k, V)$ (for some $k \in \{1, 2\}$) and closed terms $s_1, \ldots, s_\ell \in \Term^c(\Sigma \setminus \{\x_i\})$ with $\mroot(s_j) \notin \Sigma_k$ for each $1 \leq j \leq \ell$. Also, we have $\rank(s_j) < \rank(u) \leq m$ for all $1 \leq j \leq \ell$. Suppose first that $\rank(v) = 0$, and assume that $u \cong v$. Since $v$ is pure, either $v \equiv \x$ or $v$ does not contain $\x$. In the first case, since $u \cong v$, it follows that $u \equiv \x$ as well, which is impossible, since $\rank(u) > 0$.

So $v$ cannot contain $\x$, in which case $v[\x_i/\x] \equiv v$, and we must show $u[\x_i/\x] \cong v$. Since $\mroot(u) \in \Sigma_k$ and $u \cong v$, it follows that $\mroot(v) \in \Sigma_k$, so that $v \in \Term^c(\Sigma_k)$, since $v$ is pure. Since $C[s_1, \ldots, s_\ell] \equiv u \cong v$, this means that $[C[s_1, \ldots, s_\ell]]_{\cong}^k \approx_k [v]_{\cong}^k = v$, i.e. 
\[ C([s_1]_{\cong}, \ldots, [s_\ell]_{\cong}) \approx_k v. \]
\noindent Now we define the following substitution \[ \sigma : \Term^c(\Sigma)/{\cong} \to \Term^c(\Sigma_k, \Term^c(\Sigma)/{\cong}). \] Let $w \in \Term^c(\Sigma)$ be arbitrary. If $w \cong s_j$ for some $1 \leq j \leq \ell$, then we set $\sigma([w]_{\cong}) := [s_j(\x_i)]_{\cong}$. Otherwise, we set $\sigma([w]_{\cong}) := [\x]_{\cong}$. To show that $\sigma$ is well-defined, it suffices to show that if $s_{j_1} \cong s_{j_2}$ for some $1 \leq j_1, j_2 \leq \ell$, then $s_{j_1}(\x_i) \cong s_{j_2}(\x_i)$. But this follows by the induction hypothesis, since $\rank(s_{j_1}), \rank(s_{j_2}) < m$. \par

Then by Lemma \ref{closedundersubs}, we have \[ \sigma^*(C([s_1]_{\cong}, \ldots, [s_\ell]_{\cong})) \approx_k \sigma^*(v), \]
\noindent i.e. 
\[ C(\sigma([s_1]_{\cong}), \ldots, \sigma([s_\ell]_{\cong})) \approx_k v, \]
\noindent i.e.
\[ C([s_1(\x_i)]_{\cong}, \ldots, [s_\ell(\x_i)]_{\cong}) \approx_k v, \]
\noindent i.e.
\[ [C[s_1(\x_i), \ldots, s_\ell(\x_i)]]_{\cong}^k \approx_k [v]_{\cong}^k, \]
\noindent so that $C[s_1, \ldots, s_\ell](\x_i) \equiv u[\x_i/\x] \cong v$, as desired. So if $u \cong v$ and $\rank(v) = 0$, then $u[\x_i/\x] \cong v[\x_i/\x] \equiv v$. If we instead assume $u[\x_i/\x] \cong v$ in order to show $u \cong v$, then we reverse the reasoning just used, except that we define the substitution $\sigma : \Term^c(\Sigma)/{\cong} \to \Term^c(\Sigma_k, \Term^c(\Sigma)/{\cong})$ so that if $w \in \Term^c(\Sigma)$ is isomorphic to $s_j(\x_i)$ for some $1 \leq j \leq \ell$, then $\sigma([w]_{\cong}) := [s_j]_{\cong}$, and $\sigma([w]_{\cong}) := [\x]_{\cong}$ otherwise (and $\sigma$ will be well-defined by the induction hypothesis, as above). \par

Now suppose that $\rank(v) > 0$ as well, and let $v \equiv D[t_1, \ldots, t_r]$ for some $\Sigma_{k'}$-context $D \in \Term(\Sigma_{k'}, V)$ (for some $k' \in \{1, 2\}$) and closed terms $t_1, \ldots, t_r \in \Term^c(\Sigma \setminus \{\x_i\})$ with $\mroot(t_q) \notin \Sigma_{k'}$ for all $1 \leq q \leq r$. Then $\rank(t_q) < \rank(v) \leq m$ for all $1 \leq q \leq r$. Suppose first that $u \cong v$, so that $k = k'$. Then we have
\[ [C[s_1, \ldots, s_\ell]]_{\cong}^k = [u]_{\cong}^k \approx_k [v]_{\cong}^k = [D[t_1, \ldots, t_r]]_{\cong}^k, \]
\noindent i.e. \[ C([s_1]_{\cong}, \ldots, [s_\ell]_{\cong}) \approx_k D([t_1]_{\cong}, \ldots, [t_r]_{\cong}). \]
Now, the terms $s_1, \ldots, s_\ell, t_1, \ldots, t_r$ all have rank $< m$, so by the induction hypothesis, we have the following facts:
\begin{itemize}

\item For all $1 \leq j, j' \leq \ell$, $s_j \cong s_{j'}$ iff $s_j(\x_i) \cong s_{j'}(\x_i)$.

\item For all $1 \leq q, q' \leq r$, $t_q \cong t_{q'}$ iff $t_q(\x_i) \cong t_{q'}(\x_i)$.

\item For all $1 \leq j \leq \ell$ and $1 \leq q \leq r$, $s_j \cong t_q$ iff $s_j(\x_i) \cong t_q(\x_i)$.

\end{itemize}

\noindent We now define the following substitution \[ \sigma : \Term^c(\Sigma)/{\cong} \to \Term^c(\Sigma_k, \Term^c(\Sigma)/{\cong}). \] Let $w \in \Term^c(\Sigma)$ be arbitrary. If $w \cong s_j$ for some $1 \leq j \leq \ell$, then we set $\sigma([w]_{\cong}) = [s_j(\x_i)]_{\cong}$. If $w \cong t_q$ for some $1 \leq q \leq r$, then we set $\sigma([w]_{\cong}) = [t_q(\x_i)]_{\cong}$. Otherwise, we set $\sigma([w]_{\cong}) = [\x]_{\cong}$. To show that $\sigma$ is well-defined, it suffices to show that if $s_j \cong s_{j'}$ for some $1 \leq j, j' \leq \ell$, then $s_j(\x_i) \cong s_{j'}(\x_i)$, that if $t_q \cong t_{q'}$ for some $1 \leq q, q' \leq r$, then $t_q(\x_i) \cong t_{q'}(\x_i)$, and that if $s_j \cong t_q$ for some $1 \leq j \leq \ell$ and some $1 \leq q \leq r$, then $s_j(\x_i) \cong t_q(\x_i)$. But these claims follow from the above facts. \par

Then by Lemma \ref{closedundersubs}, we have 
\[ \sigma^*(C([s_1]_{\cong}, \ldots, [s_\ell]_{\cong})) \approx_k \sigma^*(D([t_1]_{\cong}, \ldots, [t_r]_{\cong})), \]
\noindent i.e.
\[ C(\sigma([s_1]_{\cong}), \ldots, \sigma([s_\ell]_{\cong})) \approx_k D(\sigma([t_1]_{\cong}), \ldots, \sigma([t_r]_{\cong})), \]
\noindent i.e.
\[ C([s_1(\x_i)]_{\cong}, \ldots, [s_\ell(\x_i)]_{\cong}) \approx_k D([t_1(\x_i)]_{\cong}, \ldots, [t_r(\x_i)]_{\cong}), \]
\noindent i.e. 
\[ C[s_1(\x_i), \ldots, s_\ell(\x_i)]_{\cong}^k \approx_k D[t_1(\x_i), \ldots, t_r(\x_i)]_{\cong}^k, \]
\noindent i.e.
\[ u[\x_i/\x] \equiv C[s_1, \ldots, s_\ell](\x_i) \cong D[t_1, \ldots, t_r](\x_i) \equiv v[\x_i/\x], \]
\noindent as desired. 

If we instead assume $u[\x_i/\x] \cong v[\x_i/\x]$ in order to show $u \cong v$, then we essentially reverse the above reasoning, by modifying the substitution defined above and showing that it is well-defined by using the other directions of the equivalences that follow by the induction hypothesis. This completes the proof.
\end{proof}

\begin{lemma}
\label{collapsinglemma}
{\em
Let $t \in \Term^c(\Sigma_1 \cup \Sigma_2 \cup \Sigma_3 \cup \{\x\})$ with $t^{\mathcal{A}} = [t]_{\cong}$ and $u^{\mathcal{A}} = [u]_{\cong}$ for every $u \in \Allaliens(t)$. Then for any $i \geq 1$, writing $t(\x_i)$ for $t[\x_i/\x]$, we have $t(\x_i)^{\mathcal{A}} = [t(\x_i)]_{\cong}$ and $v^{\mathcal{A}} = [v]_{\cong}$ for every $v \in \Allaliens(t(\x_i))$.
} 
\end{lemma}

\begin{proof}
We prove this by induction on $\rank(t)$. If $\rank(t) = 0$, then either $t$ does not contain $\x$ or $t \equiv \x$. If $t$ does not contain $\x$, then $t(\x_i) \equiv t$ and the result holds by the assumption on $t$ (or simply because $t$ has rank $0$, using Lemma \ref{pureterminterpretation}). If $t \equiv \x$, then $t(\x_i) \equiv \x_i$ and the result follows by Lemma \ref{pureterminterpretation}, since $\x_i$ is pure. \par

Now suppose that $\rank(t) = m + 1$ for some $m \geq 0$, and let $t \equiv C[s_1, \ldots, s_\ell]$ for some proper $\Sigma_k$-context $C$ (for some $k \in \{1, 2\}$) and some closed terms $s_1, \ldots, s_\ell \in \Term^c(\Sigma_1 \cup \Sigma_2 \cup \Sigma_3 \cup \{\x\})$ of rank $\leq m$, with $\mroot(s_j) \notin \Sigma_k$ for all $1 \leq j \leq \ell$. Suppose that $t^{\mathcal{A}} = [t]_{\cong}$ and $u^{\mathcal{A}} = [u]_{\cong}$ for every $u \in \Allaliens(t)$. So $C[s_1, \ldots, s_\ell]^{\mathcal{A}} = [C[s_1, \ldots, s_\ell]]_{\cong}$ and for every $1 \leq j \leq \ell$ we have $s_j^{\mathcal{A}} = [s_j]_{\cong}$ and $u_j^{\mathcal{A}} = [u_j]_{\cong}$ for every $u_j \in \Allaliens(s_j)$. Since $C[s_1, \ldots, s_\ell]^{\mathcal{A}} = [C[s_1, \ldots, s_\ell]]_{\cong}$, it follows that $C[s_1, \ldots, s_\ell]$ does not collapse to any of its alien subterms $s_1, \ldots, s_\ell$. Also, by the induction hypothesis, for every $1 \leq j \leq \ell$ it follows that $s_j(\x_i)^{\mathcal{A}} = [s_j(\x_i)]_{\cong}$ and $v^{\mathcal{A}} = [v]_{\cong}$ for every $v \in \Allaliens(s_j(\x_i))$. \par

Now we want to show that $t(\x_i)^{\mathcal{A}} = [t(\x_i)]_{\cong}$ and $v^{\mathcal{A}} = [v]_{\cong}$ for every $v \in \Allaliens(t(\x_i))$. First, note that $t(\x_i) \equiv C[s_1, \ldots, s_\ell](\x_i) \equiv C[s_1(\x_i), \ldots, s_\ell(\x_i)]$, because $C$ is a $\Sigma_k$-context (where $k \in \{1, 2\}$) and hence does not contain $\x$ (which belongs to $\Sigma_4$). \par

So, to show that $t(\x_i)^{\mathcal{A}} = [t(\x_i)]_{\cong}$, we must show that \[ C[s_1(\x_i), \ldots, s_\ell(\x_i)]^{\mathcal{A}} = [C[s_1(\x_i), \ldots, s_\ell(\x_i)]]_{\cong}. \] 

\noindent Well, we have 
\[ C[s_1(\x_i), \ldots, s_\ell(\x_i)]^{\mathcal{A}} = C^{\mathcal{A}}\left(s_1(\x_i)^{\mathcal{A}}, \ldots, s_\ell(\x_i)^{\mathcal{A}}\right) = C^{\mathcal{A}}([s_1(\x_i)]_{\cong}, \ldots, [s_\ell(\x_i)]_{\cong}), \]
\noindent with the second equality justified by the induction hypothesis. So we must now show that
\[ C^{\mathcal{A}}([s_1(\x_i)]_{\cong}, \ldots, [s_\ell(\x_i)]_{\cong}) = [C[s_1(\x_i), \ldots, s_\ell(\x_i)]]_{\cong}, \]
\noindent and to show this, we must show that $C[s_1(\x_i), \ldots, s_\ell(\x_i)]$ does not collapse to any of its alien subterms $s_1(\x_i), \ldots, s_\ell(\x_i)$. Suppose towards a contradiction that there is some $1 \leq j \leq \ell$ such that $C[s_1(\x_i), \ldots, s_\ell(\x_i)]$ collapses to $s_j(\x_i)$, which means that 
\[ [C[s_1(\x_i), \ldots, s_\ell(\x_i)]]_{\cong}^k \approx_k [s_j(\x_i)]_{\cong}, \] 
\noindent i.e.
\[ C([s_1(\x_i)]_{\cong}, \ldots, [s_j(\x_i)]_{\cong}, \ldots, [s_\ell(\x_i)]_{\cong}) \approx_k [s_j(\x_i)]_{\cong}. \]
\noindent Now we define the following substitution \[\sigma : \Term^c(\Sigma)/{\cong} \to \Term^c(\Sigma_k, \Term^c(\Sigma)/{\cong}). \] Let $w \in \Term^c(\Sigma)$ be arbitrary. If $w \cong s_j(\x_i)$ for some $1 \leq j \leq \ell$, then we set $\sigma([w]_{\cong}) = [s_j]_{\cong}$. Otherwise, we set $\sigma([w]_{\cong}) = [\x]_{\cong}$ (it does not really matter what the output of $\sigma$ is in this case; we just chose $[\x]_{\cong}$ to be concrete). To show that $\sigma$ is well-defined, it suffices to show that if $s_{j_1}(\x_i) \cong s_{j_2}(\x_i)$ for some $1 \leq j_1, j_2 \leq \ell$, then $s_{j_1} \cong s_{j_2}$. But this follows by Lemma \ref{indeterminatelemma}, since $s_{j_1}$ and $s_{j_2}$ do not contain $\x_i$. \par

Then by Lemma \ref{closedundersubs}, we have \[ \sigma^*(C([s_1(\x_i)]_{\cong}, \ldots, [s_\ell(\x_i)]_{\cong})) \approx_k \sigma^*([s_j(\x_i)]_{\cong}), \]
\noindent i.e. 
\[ C(\sigma([s_1(\x_i)]_{\cong}), \ldots, \sigma([s_\ell(\x_i)]_{\cong})) \approx_k \sigma([s_j(\x_i)]_{\cong}), \]
\noindent i.e.
\[ C([s_1]_{\cong}, \ldots, [s_\ell]_{\cong}) \approx_k [s_j]_{\cong}, \]
\noindent i.e.
\[ [C[s_1, \ldots, s_\ell]]_{\cong}^k \approx_k [s_j]_{\cong}. \] 
\noindent But this means that $C[s_1, \ldots, s_\ell]$ collapses to the alien subterm $s_j$, which contradicts the assumption on $t \equiv C[s_1, \ldots, s_\ell]$. So it must be that $C[s_1(\x_i), \ldots, s_\ell(\x_i)]$ does \emph{not} collapse to any of its alien subterms, and hence $t(\x_i)^{\mathcal{A}} = [t(\x_i)]_{\cong}$, as desired. \par

To complete the proof, we must show that $v^{\mathcal{A}} = [v]_{\cong}$ for every \[ v \in \Allaliens(t(\x_i)) = \Allaliens(C[s_1(\x_i), \ldots, s_\ell(\x_i)]). \] If $v \in \Allaliens(C[s_1(\x_i), \ldots, s_\ell(\x_i)])$, then there is some $1 \leq j \leq \ell$ such that either $v \equiv s_j(\x_i)$ or $v \in \Allaliens(s_j(\x_i))$. But then by the induction hypothesis, the desired result holds. 
\end{proof}

\noindent For the following lemma, recall that if $u$ is a closed term over $\Sigma$, then $u^{\mathcal{A}} \in \mathcal{A} = \Term^c(\Sigma)/{\cong}$, so that $u^{\mathcal{A}}$ is a $\cong$-class of closed $\Sigma$-terms. 

\begin{lemma}
\label{existencesmallerranklemma}
{\em
For any closed term $u$ over $\Sigma$, there is a closed term $u' \in u^{\mathcal{A}}$ such that $\rank(u') \leq \rank(u)$.
} 
\end{lemma}

\begin{proof}
We prove this by induction on $\rank(u)$. If $\rank(u) = 0$, then $u$ is pure and so $u^{\mathcal{A}} = [u]_{\cong}$ by Lemma \ref{pureterminterpretation}. So we have $u \in u^{\mathcal{A}}$ and obviously $\rank(u) \leq \rank(u)$. 

Now let $\rank(u) = m + 1$ for some $m \geq 0$ and suppose that the claim holds for all terms of rank $\leq m$. Let $u \equiv C[s_1, \ldots, s_\ell]$ for some proper $\Sigma_k$-context $C \in \Term(\Sigma_k, V)$ (for some $k \in \{1, 2\}$) and closed terms $s_1, \ldots, s_\ell$ over $\Sigma$ with $\mroot(s_j) \notin \Sigma_k$ for each $1 \leq j \leq \ell$. Moreover, we have $\rank(s_j) \leq m$ for all $1 \leq j \leq \ell$. Then by the induction hypothesis, for each $1 \leq j \leq \ell$ there is a closed term $s_j' \in s_j^{\mathcal{A}}$ such that $\rank(s_j') \leq \rank(s_j)$. Hence, we have $s_j^{\mathcal{A}} = [s_j']_{\cong}$ for each $1 \leq j \leq \ell$.

Now, we have \[ u^{\mathcal{A}} = [C[s_1, \ldots, s_\ell]]^{\mathcal{A}} = C^{\mathcal{A}}\left(s_1^{\mathcal{A}}, \ldots, s_\ell^{\mathcal{A}}\right) = C^{\mathcal{A}}([s_1']_{\cong}, \ldots, [s_\ell']_{\cong}). \]
\noindent Let $C(s_1', \ldots, s_\ell') \equiv C'[t_1, \ldots, t_r]$ for some $\Sigma_k$-context $C'$ and closed terms $t_1, \ldots, t_r$ over $\Sigma$ with $\mroot(t_q) \notin \Sigma_k$ for each $1 \leq q \leq r$. So for any $1 \leq q \leq r$, there is some $1 \leq j \leq \ell$ such that either $t_q \equiv s_j'$ (if $\mroot(s_j') \notin \Sigma_k$) or $t_q$ is an alien subterm of $s_j'$ (if $\mroot(s_j') \in \Sigma_k$). So for all $1 \leq q \leq r$, there is some $1 \leq j \leq \ell$ such that $\rank(t_q) \leq \rank(s_j')$, and hence \[ \bmax\{\rank(t_1), \ldots, \rank(t_r)\} \leq \bmax\{\rank(s_1'), \ldots, \rank(s_\ell')\}. \] 

\noindent To compute $u^{\mathcal{A}} = C^{\mathcal{A}}([s_1']_{\cong}, \ldots, [s_\ell']_{\cong})$, we must determine if $C(s_1', \ldots, s_\ell') \equiv C'[t_1, \ldots, t_r]$ collapses to any of its alien subterms $t_1, \ldots, t_r$ (cf. Remark 2). Suppose first that $C(s_1', \ldots, s_\ell') \equiv C'[t_1, \ldots, t_r]$ does \emph{not} collapse to any alien subterm. Then we have 
\begin{align*}
u^{\mathcal{A}}	&= C^{\mathcal{A}}([s_1']_{\cong}, \ldots, [s_\ell']_{\cong}) \\
			&= [C(s_1', \ldots, s_\ell')]_{\cong} \\
			&= [C'[t_1, \ldots, t_r]]_{\cong},
\end{align*}
\noindent so that $C'[t_1, \ldots, t_r] \in u^{\mathcal{A}}$ and 
\begin{align*}
\rank(C'[t_1, \ldots, t_r])	&= 1 + \bmax\{\rank(t_1), \ldots, \rank(t_r)\} \\
					&\leq 1 + \bmax\{\rank(s_1'), \ldots, \rank(s_\ell')\} \\
					&\leq 1 + \bmax\{\rank(s_1), \ldots, \rank(s_\ell)\} \\
					&= \rank(C[s_1, \ldots, s_\ell]) \\
					&= \rank(u),
\end{align*}
\noindent as desired.

Now suppose that $C(s_1', \ldots, s_\ell') \equiv C'[t_1, \ldots, t_r]$ \emph{does} collapse to the alien subterm $t_q$ for some $1 \leq q \leq r$. So then $u^{\mathcal{A}} = [t_q]_{\cong}$. Then $t_q \in u^{\mathcal{A}}$ and we have
\begin{align*}
\rank(t_q)	&\leq \bmax\{\rank(s_1'), \ldots, \rank(s_\ell')\} \\
		&< 1 + \bmax\{\rank(s_1), \ldots, \rank(s_\ell)\} \\
		&= \rank(u),
\end{align*}
\noindent as desired. This completes the induction and hence the proof.
\end{proof}

\noindent The following lemma asserts that substituting a function symbol applied to distinct indeterminates into the indeterminate $\x$ results in a term of rank at most \emph{one} greater than the rank of the original term. 

\begin{lemma}
\label{functionsymbollemma}
{\em
Let $g \in \Sigma_1 \cup \Sigma_2$ be a function symbol of arity $p \geq 1$. Then for any closed term $u$ over $\Sigma$, 
\[ \rank(u[g(\x_1, \ldots, \x_p)/\x]) \leq \rank(u) + 1. \]
}
\end{lemma}

\begin{proof}
We prove this by induction on $\rank(u)$. If $\rank(u) = 0$, then $u$ is pure. So either $u \equiv \x$ or $u$ does not contain $\x$. In the first case, we have $u[g(\overline{\x})/\x] \equiv g(\overline{\x})$. Note that the alien subterms of $g(\x_1, \ldots, \x_p)$ are $\x_1, \ldots, \x_p$. Then we have 
\begin{align*}
\rank(u[g(\overline{\x})/\x])	&= \rank(g(\overline{\x})) \\
					&= 1 + \bmax\{\rank(\x_1), \ldots, \rank(\x_p)\} \\
					&= 1 + 0 \\
					&= 0 + 1 \\
					&= \rank(u) + 1,
\end{align*}
\noindent as required. If $u$ does not contain $\x$, then we have $u[g(\overline{\x})/\x] \equiv u$, and we clearly have $\rank(u) \leq \rank(u) + 1$.

Now suppose that $u$ has rank $m + 1$ for some $m \geq 0$, so that $u \equiv C[s_1, \ldots, s_\ell]$ for some proper $\Sigma_k$-context $C$ (for some $k \in \{1, 2\}$) and closed terms $s_1, \ldots, s_\ell$ over $\Sigma$ with $\mroot(s_j) \notin \Sigma_k$ for each $1 \leq j \leq \ell$. So for every $1 \leq j \leq \ell$, we have $\rank(s_j) \leq m$, and hence we have $\rank(s_j[g(\overline{\x})/\x]) \leq \rank(s_j) + 1$ by the induction hypothesis. 

Suppose first that $g \notin \Sigma_k$. Then for each $1 \leq j \leq \ell$, we have $\mroot(s_j[g(\overline{\x})/\x]) \notin \Sigma_k$. So the alien subterms of \[ u[g(\overline{\x})/\x] \equiv C[s_1, \ldots, s_\ell][g(\overline{\x})/\x] \equiv C(s_1[g(\overline{\x})/\x], \ldots, s_\ell[g(\overline{\x})/\x]) \] are exactly $s_1[g(\overline{\x})/\x], \ldots, s_\ell[g(\overline{\x})/\x]$, so that we may write \[ u[g(\overline{\x})/\x] \equiv C[s_1[g(\overline{\x})/\x], \ldots, s_\ell[g(\overline{\x})/\x]]. \]
\noindent Then we have
\begin{align*}
\rank(u[g(\overline{\x})/\x])	&= \rank(C[s_1[g(\overline{\x})/\x], \ldots, s_\ell[g(\overline{\x})/\x]]) \\
					&= 1 + \bmax\{\rank(s_1[g(\overline{\x})/\x]), \ldots, \rank(s_\ell[g(\overline{\x})/\x])\} \\
					&\leq 1 + \bmax\{\rank(s_1) + 1, \ldots, \rank(s_\ell) + 1\} \\
					&= 1 + \bmax\{\rank(s_1), \ldots, \rank(s_\ell)\} + 1 \\
					&= \rank(C[s_1, \ldots, s_\ell]) + 1\\
					&= \rank(u) + 1,
\end{align*}
\noindent as required. 

Now suppose $g \in \Sigma_k$. Then there \emph{may} be some $1 \leq j \leq \ell$ such that $s_j \equiv \x$, in which case $s_j[g(\overline{\x})/\x] \equiv g(\overline{\x})$, so that $s_j[g(\overline{\x})/\x]$ would \emph{not} be an alien subterm of $u[g(\overline{\x})/\x] \equiv C(s_1[g(\overline{\x})/\x], \ldots, s_\ell[g(\overline{\x})/\x])$, since $g \in \Sigma_k$ and $C$ is a $\Sigma_k$-context. So let $C(s_1[g(\overline{\x})/\x], \ldots, s_\ell[g(\overline{\x})/\x]) \equiv C'[t_1, \ldots, t_r]$ for some proper $\Sigma_k$-context $C'$ and closed terms $t_1, \ldots, t_r$ over $\Sigma$ with $\mroot(t_q) \notin \Sigma_k$ for each $1 \leq q \leq r$. So for any $1 \leq q \leq r$, there is some $1 \leq j \leq \ell$ such that either $s_j \not\equiv \x$ and $t_q \equiv s_j[g(\overline{\x})/\x]$, or $s_j \equiv \x$ and $t_q \equiv \x_i$ for some $1 \leq i \leq p$. So for any $1 \leq q \leq r$, either $\rank(t_q) = 0$ or there is some $1 \leq j \leq \ell$ such that $\rank(t_q) =\ rank(s_j[g(\overline{\x})/\x]) \leq \rank(s_j) + 1$. So we have 
\[ \bmax\{\rank(t_1), \ldots, \rank(t_r)\} \leq \bmax\{\rank(s_1) + 1, \ldots, \rank(s_\ell) + 1\}. \] Hence, we have
\begin{align*}
\rank(u[g(\overline{\x})/\x])	&= \rank(C'[t_1, \ldots, t_r]) \\
					&= 1 + \bmax\{\rank(t_1), \ldots, \rank(t_r)\} \\
					&\leq 1 + \bmax\{\rank(s_1) + 1, \ldots, \rank(s_\ell) + 1\} \\
					&= 1 + \bmax\{\rank(s_1), \ldots, \rank(s_\ell)\} + 1 \\
					&= \rank(C[s_1, \ldots, s_\ell]) + 1\\
					&= \rank(u) + 1,
\end{align*}
\noindent as required. This completes the induction and hence the proof.
\end{proof}

\noindent The following lemma will be crucial in the proof of our main theorem. It essentially enumerates all possible interpretations that the term $u[g(\x_1, \ldots, \x_p)/\x]$ could have in $\mathcal{A}$, where $u$ is a closed term over $\Sigma$ and $g$ is a function symbol of $\Sigma_1 \cup \Sigma_2$ of arity $p \geq 1$. 

\begin{lemma}
\label{mostimportantlemma}
{\em
Let $g \in \Sigma_1 \cup \Sigma_2$ be a function symbol of arity $p \geq 1$, and let $g(\overline{\x}) \equiv g(\x_1, \ldots, \x_p)$. Then for any closed term $u$ over $\Sigma$ and any $v \in \{u\} \cup \Allaliens(u)$, one of the following is true:
\begin{enumerate}

\item $\rank(v) > 0$ and $v \equiv C[s_1, \ldots, s_\ell]$ for some $\Sigma_k$-context $C$ (for some $k \in \{1, 2\}$) and closed terms $s_1, \ldots, s_\ell$ over $\Sigma$ such that for any $1 \leq j \leq \ell$ we have $\mroot(s_j) \notin \Sigma_k$ and $s_j[g(\overline{\x})/\x]^{\mathcal{A}} = [s_j']_{\cong}$ and \[ v[g(\overline{\x})/\x]^{\mathcal{A}} = [C(s_1', \ldots, s_\ell')]_{\cong}. \]

\item $v[g(\overline{\x})/\x]^{\mathcal{A}} = [w]_{\cong}$ for some closed term $w$ over $\Sigma_1 \cup \Sigma_2 \cup \Sigma_3$. 

\item $v[g(\overline{\x})/\x]^{\mathcal{A}} = [g(\overline{\x})]_{\cong}$.

\item $v[g(\overline{\x})/\x]^{\mathcal{A}} = [\x_i]_{\cong}$ for some $i \geq 1$.

\item $\rank(v) > 0$ and $\mroot(v) \in \Sigma_k$ for some $k \in \{1, 2\}$ and $v[g(\overline{\x})/\x]^{\mathcal{A}} = [D(t_1', \ldots, t_r')]_{\cong}$ for some proper $\Sigma_{k'}$-context $D$ (for $k' \neq k \in \{1, 2\}$) such that for every $1 \leq q \leq r$, there is a term $t_q$ of rank $\leq \rank(v) - 2$ with $\mroot(t_q) \notin \Sigma_{k'}$ and $t_q[g(\overline{\x})/\x]^{\mathcal{A}} = [t_q']_{\cong}$. 
\end{enumerate}
}
\end{lemma}

\begin{proof}
We prove this by induction on $\rank(u)$. If $\rank(u) = 0$, then $u$ is pure. Since $\Allaliens(u) = \emptyset$ in this case, it suffices to show that one of the statements in the lemma is true for just $u$. Since $u$ is pure, either $u \equiv \x$, or $u$ does not contain $\x$. In the first case, we have $u[g(\overline{\x})/\x] \equiv g(\overline{\x})$, so that \[ u[g(\overline{\x})/\x]^{\mathcal{A}} = g(\x_1, \ldots, \x_p)^{\mathcal{A}} = g^{\mathcal{A}}\left(\x_1^{\mathcal{A}}, \ldots, \x_p^{\mathcal{A}}\right) = g^{\mathcal{A}}([\x_1]_{\cong}, \ldots, [\x_p]_{\cong}), \] where the last equality holds by Lemma \ref{pureterminterpretation}, because $\x_i$ is pure for each $1 \leq i \leq p$. Now we must consider whether $g(\x_1, \ldots, \x_p)$ collapses to one of the alien subterms $\x_1, \ldots, \x_p$. If it does not, then we have \[ u[g(\overline{\x})/\x]^{\mathcal{A}} = g^{\mathcal{A}}([\x_1]_{\cong}, \ldots, [\x_p]_{\cong}) = [g(\overline{\x})]_{\cong}, \] so that statement (3) is satisfied. Otherwise, we have $u[g(\overline{\x})/\x]^{\mathcal{A}} = [\x_i]_{\cong}$ for some $1 \leq i \leq p$, so that statement (4) is satisfied. 

If $u$ does \emph{not} contain $\x$ and hence is a pure term over $\Sigma \setminus \{\x\}$, then we have $u[g(\overline{\x})/\x] \equiv u$ and hence $u[g(\overline{\x})/\x]^{\mathcal{A}} = u^{\mathcal{A}} = [u]_{\cong}$ by Lemma \ref{pureterminterpretation}. If $u \equiv \x_i$ for some $i \geq 1$, then statement (4) is satisfied. Otherwise, statement (2) is satisfied. This completes the base case. 

Now let $m \geq 0$ and suppose that $\rank(u) = m + 1$ and that the result holds for any closed term over $\Sigma$ of rank $\leq m$. Let $u \equiv C[s_1, \ldots, s_\ell]$, where $C$ is a proper $\Sigma_k$-context (for some $k \in \{1, 2\}$) and $s_1, \ldots, s_\ell$ are closed terms over $\Sigma$ of rank $\leq m$ such that $\mroot(s_j) \notin \Sigma_k$ for each $1 \leq j \leq \ell$. Moreover, for any $1 \leq j \leq \ell$, suppose that $s_j[g(\overline{\x})/\x]^{\mathcal{A}} = [s_j']_{\cong}$. Then we have
\begin{align*}
u[g(\overline{\x})/\x]^{\mathcal{A}}	&= C[s_1, \ldots, s_\ell][g(\overline{\x})/\x]^{\mathcal{A}} \\
							&= C(s_1[g(\overline{\x})/\x], \ldots, s_\ell[g(\overline{\x})/\x])^{\mathcal{A}} \\
							&= C^{\mathcal{A}}\left(s_1[g(\overline{\x})/\x]^{\mathcal{A}}, \ldots, s_\ell[g(\overline{\x})/\x]^{\mathcal{A}}\right) \\ 
							&= C^{\mathcal{A}}([s_1']_{\cong}, \ldots, [s_\ell']_{\cong}).
\end{align*}
Before we show that one of the statements (1) -- (5) holds for $u$, we first show that one of the statements (1) - (5) holds for any $v \in \Allaliens(u)$. If $v \in \Allaliens(u)$, then there is some $1 \leq j \leq \ell$ such that either $v \equiv s_j$ or $v \in \Allaliens(s_j)$. But $\rank(s_j) \leq m$, so by the induction hypothesis, the desired result follows. 

Now, to show that one of the statements (1) -- (5) holds for $u$ itself, we must consider whether $C(s_1', \ldots, s_\ell')$ collapses to any of its alien subterms. If not, then we have \[ u[g(\overline{\x})/\x]^{\mathcal{A}} = C^{\mathcal{A}}([s_1']_{\cong}, \ldots, [s_\ell']_{\cong}) = [C(s_1', \ldots, s_\ell')]_{\cong}, \] and so statement (1) is satisfied for $u$.
 
Otherwise, suppose that $C(s_1', \ldots, s_\ell')$ \emph{does} collapse to one of its alien subterms, say $s$. Then there is some $1 \leq j \leq \ell$ such that either $s \equiv s_j'$ (if $\mroot(s_j') \notin \Sigma_k$), or $s$ is an alien subterm of $s_j'$ (if $\mroot(s_j') \in \Sigma_k$). 
\begin{itemize}
\item Suppose first that $s \equiv s_j'$, so that $\mroot(s_j') \notin \Sigma_k$. Then we have \[ u[g(\overline{\x})/\x]^{\mathcal{A}} = [s]_{\cong} = [s_j']_{\cong} = s_j[g(\overline{\x})/\x]^{\mathcal{A}}. \] 
By the induction hypothesis, we know that one of the statements (1) -- (5) is true for $s_j$. 
\begin{itemize}
\item Suppose that statement (1) is true for $s_j$. So $\rank(s_j) > 0$ and $s_j \equiv D[t_1, \ldots, t_r]$ for some proper $\Sigma_{k'}$-context $D$ (for some $k' \in \{1, 2\}$) and closed terms $t_1, \ldots, t_r$ over $\Sigma$ such that $\mroot(t_q) \notin \Sigma_{k'}$ and $t_q[g(\overline{\x})/\x]^{\mathcal{A}} = [t_q']_{\cong}$ for each $1 \leq q \leq r$, and \[ s_j[g(\overline{\x})/\x]^{\mathcal{A}} = [D(t_1', \ldots, t_r')]_{\cong}. \] Since $s_j[g(\overline{\x})/\x]^{\mathcal{A}} = [s_j']_{\cong}$, we may assume without loss of generality that \[ s_j' \equiv D(t_1', \ldots, t_r'). \] Then we have \[ u[g(\overline{\x})/\x]^{\mathcal{A}} = s_j[g(\overline{\x})/\x]^{\mathcal{A}} = [s_j']_{\cong} = [D(t_1', \ldots, t_r')]_{\cong}. \] 

\noindent Since $\mroot(s_j') \notin \Sigma_k$ in this case, it follows that $k \neq k'$, since $\mroot(s_j') = \mroot(D) \in \Sigma_{k'}$. Now we show that statement (5) is satisfied for $u$. By the preceding discussion, it remains to show that $\rank(t_q) \leq \rank(u) - 2$ for each $1 \leq q \leq r$. But this follows because we have $\rank(t_q) < \rank(s_j) < \rank(u)$. 

\item If one of the statements (2) -- (4) is true for $s_j$, then since $u[g(\overline{\x})/\x]^{\mathcal{A}} = s_j[g(\overline{\x})/\x]^{\mathcal{A}}$, the corresponding statement is also true for $u$. 

\item Lastly, we show that it is \emph{not} possible for statement (5) to be true of $s_j$. If it \emph{were} true of $s_j$, then it would follow that $s_j[g(\overline{\x})/\x]^{\mathcal{A}} = [s_j']_{\cong} = [w]_{\cong}$ for some closed term $w$ with $\mroot(w) \in \Sigma_k$ (since $\mroot(s_j) \notin \Sigma_k$). But this is impossible, since then $s_j' \cong w$ and $\mroot(s_j') \notin \Sigma_k$ but $\mroot(w) \in \Sigma_k$. 
\end{itemize}
\item Now suppose that $s$ is an alien subterm of $s_j'$ for some $1 \leq j \leq \ell$, so that $\mroot(s_j') \in \Sigma_k$ and $C(s_1', \ldots, s_\ell')$ collapses to $s$. Then we have $u[g(\overline{\x})/\x]^{\mathcal{A}} = [s]_{\cong}$. Again, we know by the induction hypothesis that one of the statements (1) -- (5) is true for $s_j$.
\begin{itemize}
\item Suppose first that statement (1) is true for $s_j$. Then we would have $s_j[g(\overline{\x})/\x]^{\mathcal{A}} = [s_j']_{\cong} = [w]_{\cong}$ for some closed term $w$ with $\mroot(w) \notin \Sigma_k$, since $\mroot(s_j) \notin \Sigma_k$. But then we would have $s_j' \cong w$ and $\mroot(s_j') \in \Sigma_k$ and $\mroot(w) \notin \Sigma_k$, which is impossible. So statement (1) cannot be true for $s_j$ in this case. 

\item If statement (2) is true for $s_j$, then we would have $s_j[g(\overline{\x})/\x]^{\mathcal{A}} = [s_j']_{\cong} = [w]_{\cong}$ for some closed term $w$ over $\Sigma_1 \cup \Sigma_2 \cup \Sigma_3$. Without loss of generality, we may assume that $s_j' \equiv w$, so that $s_j'$ is a closed term over $\Sigma_1 \cup \Sigma_2 \cup \Sigma_3$. Since $s$ is an alien subterm of $s_j'$, it follows that $s$ is also a closed term over $\Sigma_1 \cup \Sigma_2 \cup \Sigma_3$, so that statement (2) is also true for $u$. 

\item If statement (3) is true for $s_j$, then we would have $s_j[g(\overline{\x})/\x]^{\mathcal{A}} = [s_j']_{\cong} = [g(\overline{\x})]_{\cong}$. Again, we may assume without loss of generality that $s_j' \equiv g(\overline{\x})$. Since $s$ is an alien subterm of $s_j' \equiv g(\x_1, \ldots, \x_p)$, it follows that $s \equiv \x_i$ for some $1 \leq i \leq p$. Then we have $u[g(\overline{\x})/\x]^{\mathcal{A}} = [s]_{\cong} = [\x_i]_{\cong}$ and thus statement (4) is true for $u$. 

\item If statement (4) is true for $s_j$, then we would have $s_j[g(\overline{\x})/\x]^{\mathcal{A}} = [s_j']_{\cong} = [\x_i]_{\cong}$ for some $i \geq 1$. Again, we may assume without loss of generality that $s_j' \equiv \x_i$. But then $s_j'$ has no alien subterms, contradicting the assumption that $s$ is an alien subterm of $s_j'$. So statement (4) cannot be true for $s_j$ in this case. 

\item Finally, suppose that statement (5) is true for $s_j$. Then we would have $s_j[g(\overline{\x})/\x]^{\mathcal{A}} = [s_j']_{\cong} = [D(t_1', \ldots, t_r')]_{\cong}$, where $D$ is a proper $\Sigma_k$-context (since $\mroot(s_j) \notin \Sigma_k$) and for every $1 \leq q \leq r$, there is a term $t_q$ of rank $\leq \rank(s_j) - 2$ such that $\mroot(t_q) \notin \Sigma_k$ and $t_q[g(\overline{\x})/\x]^{\mathcal{A}} = [t_q']_{\cong}$. Again, we may assume without loss of generality that $s_j' \equiv D(t_1', \ldots, t_r')$. Since $s$ is an alien subterm of $s_j'$, it follows that there is some $1 \leq q \leq r$ such that either $s \equiv t_q'$ (if $\mroot(t_q') \notin \Sigma_k$) or $s$ is an alien subterm of $t_q'$ (if $\mroot(t_q') \in \Sigma_k$). In other words, there is some $1 \leq q \leq r$ such that $s \in \{t_q'\} \cup \Allaliens(t_q')$. 

Since $t_q' \in [t_q']_{\cong} = t_q[g(\overline{\x})/\x]^{\mathcal{A}}$, by Lemma \ref{existencesmallerranklemma} we may assume without loss of generality that $\rank(t_q') \leq \rank(t_q[g(\overline{\x})/\x]) \leq \rank(t_q) + 1$, where the latter inequality comes from Lemma \ref{functionsymbollemma}. So we have \[ \rank(t_q') \leq \rank(t_q) + 1 \leq (\rank(s_j) - 2) + 1 = \rank(s_j) - 1 < m. \] So by the induction hypothesis, every term in $\{t_q'\} \cup \Allaliens(t_q')$ satisfies one of the statements (1) -- (5), and hence $s$ satisfies one of these statements, as desired.
\end{itemize} 
\end{itemize}
This completes the induction and hence the proof.
\end{proof}

\noindent We can now finally prove that $M_n := (\T_1 + \T_2)(\y_1, \ldots, \y_n)$, the free model of $\T_1 + \T_2$ on $n$ generators $\y_1, \ldots, \y_n$, has trivial isotropy group. 

\begin{theorem}
\label{maintheorem}
{\em Let $\T_1, \T_2$ be equational theories over respective disjoint signatures $\Sigma_1, \Sigma_2$, and suppose for each $i \in \{1, 2\}$ that $\T_i$ has at least one function symbol that is neither constant nor a projection in $\T_i$. Then for any $n \geq 0$, \[ G_{\T_1 + \T_2}(M_n) = \{ [\x] \}, \] where $M_n$ is the free model of $\T_1 + \T_2$ on $n$ generators $\y_1, \ldots, \y_n$. 

In other words, the logical isotropy group of the free model of $\T_1 + \T_2$ on $n$ generators is trivial. 
}
\end{theorem}

\begin{proof}
Let $t \in \Term^c(\Sigma_1 \cup \Sigma_2 \cup \Sigma_3 \cup \{\x\})$ with $[t] \in G_{\T_1 + \T_2}(M_n)$. We want to show that $[t] = [\x]$, i.e. that $t \sim_{M_n, \x} \x$. If $\T_1 + \T_2$ is \emph{trivial}, meaning that it proves the equation $x = y$ for distinct variables $x, y$, then it easily follows that 
\[ (\T_1 + \T_2)(\x, \y_1, \ldots, \y_n) \vdash t = \x, \] which means $t \sim_{M_n, \x} \x$. 

So suppose that $\T_1 + \T_2$ is \emph{not} trivial, i.e. that $\T_1 + \T_2$ does \emph{not} prove the equation $x = y$ for distinct variables $x, y$. Without loss of generality, we may assume that $t$ satisfies statements (3) and (4) from Lemma \ref{canonicaltermlemma}. Because if $t$ did \emph{not} satisfy these statements, then by Lemma \ref{canonicaltermlemma} we could find a closed term $t' \in \Term^c(\Sigma_1 \cup \Sigma_2 \cup \Sigma_3 \cup \{\x\})$ satisfying statements (2) -- (4) of Lemma \ref{canonicaltermlemma}. Then since $[t]$ is an element of the logical isotropy group and $t \sim_{M_n, \x} t'$, it would follow that $[t'] = [t]$ is an element of the logical isotropy group as well. So if we proved $t' \sim_{M_n, \x} \x$, we would have $t \sim_{M_n, \x} t' \sim_{M_n, \x} \x$, as desired. 

So for every $u \in \{t\} \cup \Allaliens(t)$ we have $u^{\mathcal{A}} = [u]_{\cong}$, and $t$ is not isomorphic to any closed term over $\Sigma_1 \cup \Sigma_2 \cup \Sigma_3 \cup \{\x\}$ of strictly smaller rank. Given that $t$ satisfies statements (3) and (4) from Lemma \ref{canonicaltermlemma}, we will actually prove the stronger claim that $t \equiv \x$. 

So suppose towards a contradiction that $t \not\equiv \x$. Since $\T_1 + \T_2$ is not trivial and $[t] \in G_{\T_1 + \T_2}(M_n)$, it follows by \cite[Lemma 2.2.56]{thesis} that $t$ must contain $\x$, so that $\rank(t) > 0$. So we have $t \equiv C[s_1, \ldots, s_\ell]$ for some proper $\Sigma_k$-context $C$ (for some $k \in \{1, 2\}$) and closed terms $s_1, \ldots, s_\ell$ over $\Sigma_1 \cup \Sigma_2 \cup \Sigma_3 \cup \{\x\}$ with $\mroot(s_j) \notin \Sigma_k$ for each $1 \leq j \leq \ell$.

By assumption on $\T_1$ and $\T_2$, for $k' \in \{1, 2\}$ with $k \neq k'$ there is a function symbol $f$ of $\Sigma_{k'}$ of arity $m \geq 1$ such that $f$ is neither constant nor a projection in $\T_{k'}$. Since $[t]$ is an element of the logical isotropy group, we know that $[t]$ commutes generically with $f$. More precisely, writing $\sim_{M_n, m}$ for $\sim_{M_n, \{1, \ldots, m\}}$, we have \[ t[f(\x_1, \ldots, \x_m)/\x] \sim_{M_n, m} f(t(\x_1), \ldots, t(\x_m)), \] where $t(\x_i) \equiv t[\x_i/\x]$ for each $1 \leq i \leq m$. By the first fact of Lemma \ref{infinitecongruencelemma}, we then have \[ t[f(\x_1, \ldots, \x_m)/\x] \sim_{M_n, \infty} f(t(\x_1), \ldots, t(\x_m)), \]  and then by \cite[Theorem 3.1]{Combining}, it follows that \[ t[f(\x_1, \ldots, \x_m)/\x]^{\mathcal{A}} = f(t(\x_1), \ldots, t(\x_m))^{\mathcal{A}}. \] First we calculate the right side of this equation. So consider the term $f(t(\x_1), \ldots, t(\x_m))$. Since $\mroot(t) \in \Sigma_k$ and $f \notin \Sigma_k$, the alien subterms of this term are $t(\x_1), \ldots, t(\x_m)$. We claim that $f(t(\x_1), \ldots, t(\x_m))$ does \emph{not} collapse to any of these alien subterms. Suppose towards a contradiction that $f(t(\x_1), \ldots, t(\x_m))$ \emph{did} collapse to $t(\x_i)$ for some $1 \leq i \leq m$. Then this would mean that \[ [f(t(\x_1), \ldots, t(\x_m))]_{\cong}^{k'} \approx_{k'} [t(\x_i)]_{\cong}^{k'}, \] i.e. \[ f([t(\x_1)]_{\cong}, \ldots, [t(\x_m)]_{\cong}) \approx_{k'} [t(\x_i)]_{\cong}. \]

\noindent By Lemma \ref{distinctindeterminateslemma} below, we know that $t(\x_i) \not\cong t(\x_j)$ for any $1 \leq i \neq j \leq m$, and hence $[t(\x_i)]_{\cong} \neq [t(\x_j)]_{\cong}$. Then by Lemma \ref{complexlemma}, we obtain 
\[ \T_{k'} \vdash f(y_1, \ldots, y_m) = y_i \] for distinct variables $y_1, \ldots, y_m \in V$, which contradicts the assumption that $f$ is \emph{not} a projection in $\T_{k'}$. So $f(t(\x_1), \ldots, t(\x_m))$ does \emph{not} collapse to any alien subterm, as desired. Furthermore, since $t^{\mathcal{A}} = [t]_{\cong}$ and $v^{\mathcal{A}} = [v]_{\cong}$ for every $v \in \Allaliens(t)$, it follows by Lemma \ref{collapsinglemma} that $t(\x_i)^{\mathcal{A}} = [t(\x_i)]_{\cong}$ for every $1 \leq i \leq m$. Thus, we have
\begin{align*}
f(t(\x_1), \ldots, t(\x_m))^{\mathcal{A}}	&= f^{\mathcal{A}}\left(t(\x_1)^{\mathcal{A}}, \ldots, t(\x_m)^{\mathcal{A}}\right) \\
							&= f^{\mathcal{A}}([t(\x_1)]_{\cong}, \ldots, [t(\x_m)]_{\cong}) \\
							&= [f(t(\x_1), \ldots, t(\x_m))]_{\cong},
\end{align*}

\noindent where the last equality holds because $f(t(\x_1), \ldots, t(\x_m))$ does not collapse to any of its alien subterms, as just shown. 

\begin{lemma}
\label{distinctindeterminateslemma}
{\em
For any $1 \leq i \neq j \leq m$, we have $t(\x_i) \not\cong t(\x_j)$.
}
\end{lemma}

\begin{proof}
Suppose to the contrary that we had $t(\x_i) \cong t(\x_j)$ for some distinct $1 \leq i, j \leq m$. Then by \cite[Lemma 3.5]{Combining}, we would have $t(\x_i) \sim_{M_n, \infty} t(\x_j)$, and hence $t(\x_i) \sim_{M_n, \{i, j\}} t(\x_j)$ by Lemma \ref{infinitecongruencelemma}, meaning that
\[ (\T_1 + \T_2)(\y_1, \ldots, \y_n, \x_i, \x_j) \vdash t(\x_i) = t(\x_j). \] Since $[t]$ is an element of the logical isotropy group, it follows that $[t]$ is invertible, so that there is some closed term $t^{-1}$ over $\Sigma_1 \cup \Sigma_2 \cup \Sigma_3 \cup \{\x\}$ with
\[ (\T_1 + \T_2)(\y_1, \ldots, \y_n, \x) \vdash t[t^{-1}/\x] = \x = t^{-1}[t/\x]. \] From this and the previous equation, it then easily follows that
\[ (\T_1 + \T_2)(\y_1, \ldots, \y_n, \x, \x_i, \x_j) \vdash \x_i = \x_j. \] By \cite[Lemma 5.1.30]{thesis} and \cite[Theorem 10]{Horn}, we then obtain
\[ \T_1 + \T_2 \vdash y_i = y_j \] for distinct variables $y_i, y_j$ (since $\x_i \not\equiv \x_j$), which contradicts the assumption that $\T_1 + \T_2$ is not trivial. This completes the proof.
\end{proof}

\noindent We now have \[ t[f(\x_1, \ldots, \x_m)/\x]^{\mathcal{A}} = [f(t(\x_1), \ldots, t(\x_m))]_{\cong}, \]
\noindent so we turn to calculating the left side of this equation. By Lemma \ref{mostimportantlemma}, one of the following statements is true (recall that $t \equiv C[s_1, \ldots, s_\ell]$ for $C$ a $\Sigma_k$-context):
\begin{enumerate}

\item $t[f(\overline{\x})/\x]^{\mathcal{A}} = [C(s_1', \ldots, s_\ell')]_{\cong}$, where $s_j[f(\overline{\x})/\x]^{\mathcal{A}} = [s_j']_{\cong}$ for all $1 \leq j \leq \ell$.

\item $t[f(\overline{\x})/\x]^{\mathcal{A}} = [w]_{\cong}$, for some closed term $w$ over $\Sigma_1 \cup \Sigma_2 \cup \Sigma_3$. 

\item $t[f(\overline{\x})/\x]^{\mathcal{A}} = [f(\overline{\x})]_{\cong}$.

\item $t[f(\overline{\x})/\x]^{\mathcal{A}} = [\x_i]_{\cong}$ for some $i \geq 1$.

\item $t[f(\overline{\x})/\x]^{\mathcal{A}} = [D(t_1', \ldots, t_r')]_{\cong}$, where $D$ is a proper $\Sigma_{k'}$-context (for $k' \neq k \in \{1, 2\}$) and for every $1 \leq q \leq r$, there is a term $t_q$ of rank $\leq \rank(t) - 2$ such that $\mroot(t_q) \notin \Sigma_{k'}$ and $t_q[f(\overline{\x})/\x]^{\mathcal{A}} = [t_q']_{\cong}$. 
\end{enumerate}

\noindent To obtain a contradiction and complete the proof, we will show that, in fact, \emph{none} of these statements can be true, contrary to Lemma \ref{mostimportantlemma}. Suppose first that (1) is true. Then we would have 

\[ [C(s_1', \ldots, s_l')]_{\cong} = t[f(\overline{\x})/\x]^{\mathcal{A}} = [f(t(\x_1), \ldots, t(\x_m))]_{\cong}, \]

\noindent so that $C(s_1', \ldots, s_l') \cong f(t(\x_1), \ldots, t(\x_m))$. But this is impossible, because $\mroot(C) \in \Sigma_k$ and $f \notin \Sigma_k$. 

Now suppose that statement (2) is true. Then we have $t[f(\overline{\x})/\x]^{\mathcal{A}} = [w]_{\cong}$, for some closed term $w$ over $\Sigma_1 \cup \Sigma_2 \cup \Sigma_3$. If $\mroot(w) \notin \Sigma_{k'}$, then we already have a contradiction, because we would have \[ [w]_{\cong} = t[f(\overline{\x})/\x]^{\mathcal{A}} = [f(t(\x_1), \ldots, t(\x_m))]_{\cong}, \] so that $w \cong f(t(\x_1), \ldots, t(\x_m))$, which is impossible because $\mroot(f) \in \Sigma_{k'}$. So we must have $\mroot(w) \in \Sigma_{k'}$. Suppose first that $\rank(w) = 0$, so that $w$ is a pure $\Sigma_{k'}$-term. Then we have \[ f(t(\x_1), \ldots, t(\x_m)) \cong w, \] which implies that \[ [f(t(\x_1), \ldots, t(\x_m))]_{\cong}^{k'} \approx_{k'} [w]_{\cong}^{k'}, \] so that
\[ f([t(\x_1)]_{\cong}, \ldots, [t(\x_m)]_{\cong}) \approx_{k'} w. \] Then by Lemma \ref{complexlemma} we would obtain \[ \T_{k'} \vdash f(y_1, \ldots, y_m) = w \] for distinct variables $y_1, \ldots, y_m$, because $t(\x_i) \not\cong t(\x_j)$ for all $1 \leq i \neq j \leq m$ by Lemma \ref{distinctindeterminateslemma}. But $w$ is a closed $\Sigma_{k'}$-term (and hence in particular does not contain any of the variables $y_1, \ldots, y_m$), which contradicts the assumption that $f$ is not constant in $\T_{k'}$. 

Now suppose that $\rank(w) > 0$ (and $\mroot(w) \in \Sigma_{k'}$), and let $w \equiv D[v_1, \ldots, v_r]$ for some proper $\Sigma_{k'}$-context $D$ and closed terms $v_1, \ldots, v_r$ over $\Sigma_1 \cup \Sigma_2 \cup \Sigma_3$ such that $\mroot(v_q) \notin \Sigma_{k'}$ for all $1 \leq q \leq r$. Then we have \[ [D[v_1, \ldots, v_r]]_{\cong} = [w]_{\cong} = [f(t(\x_1), \ldots, t(\x_m))]_{\cong} \] and hence \[ f(t(\x_1), \ldots, t(\x_m)) \cong D[v_1, \ldots, v_r]. \] This implies that
\[ [f(t(\x_1), \ldots, t(\x_m))]_{\cong}^{k'} \approx_{k'} [D[v_1, \ldots, v_r]]_{\cong}^{k'}, \] which means that
\[ f([t(\x_1)]_{\cong}, \ldots, [t(\x_m)]_{\cong}) \approx_{k'} D([v_1]_{\cong}, \ldots, [v_r]_{\cong}). \]
\noindent By Lemma \ref{isotropytermnotconstant} below, we know that for any $1 \leq i \leq m$ there is no $1 \leq q \leq r$ such that $t(\x_i) \cong v_q$, since $v_q$ is a closed term over $\Sigma_1 \cup \Sigma_2 \cup \Sigma_3$. Now suppose that the terms $v_1, \ldots, v_r$ can be partitioned into $p \geq 1$ $\cong$-classes, and for any $1 \leq q \leq r$, let $1 \leq p_q \leq p$ be such that $v_q$ belongs to the $p_q^{\text{th}}$ such $\cong$-class. Let $y_1, \ldots, y_m, z_1, \ldots, z_p$ be pairwise distinct variables. Then by Lemma \ref{complexlemma} (and Lemma \ref{distinctindeterminateslemma}), the equality \[ [f(t(\x_1), \ldots, t(\x_m))]_{\cong}^{k'} \approx_{k'} [D[v_1, \ldots, v_r]]_{\cong}^{k'} \] implies
\[ \T_{k'} \vdash f(y_1, \ldots, y_m) = D\left(z_{p_1}, \ldots, z_{p_r}\right). \] But $D(z_{p_1}, \ldots, z_{p_r})$ does not contain any of the variables $y_1, \ldots, y_m$, so this contradicts the assumption that $f$ is not constant in $\T_{k'}$. This finishes the argument that statement (2) cannot be true for $t$. 

\begin{lemma}
\label{isotropytermnotconstant}
{\em
For any $1 \leq i \leq m$, there is no closed term $v$ over $\Sigma_1 \cup \Sigma_2 \cup \Sigma_3$ such that $t(\x_i) \cong v$. 
}
\end{lemma}

\begin{proof}
Suppose towards a contradiction that there is some $1 \leq i \leq m$ and some closed term $v$ over $\Sigma_1 \cup \Sigma_2 \cup \Sigma_3$ with $t(\x_i) \cong v$. Since neither $t$ nor $v$ contains $\x_i$, it follows by Lemma \ref{indeterminatelemma} that $t \cong v$ (since $v$ does not contain $\x$). Then by \cite[Lemma 3.5]{Combining}, it follows that $t \sim_{M_n, \infty} v$, which in turn implies that $t \sim_{M_n, \x} v$ by Lemma \ref{infinitecongruencelemma}. So $(\T_1 + \T_2)(\y_1, \ldots, \y_n, \x) \vdash t = v$, but this is impossible by \cite[Lemma 2.2.56]{thesis}, since $\T_1 + \T_2$ is not trivial and $[t]$ is an element of the logical isotropy group and $v$ does not contain $\x$. 
\end{proof}

Suppose now that statement (3) is true for $t$, so that $t[f(\overline{\x})/\x]^{\mathcal{A}} = [f(\overline{\x})]_{\cong}$. Then we have \[ [f(\x_1, \ldots, \x_m)]_{\cong} = t[f(\overline{\x})/\x]^{\mathcal{A}} = [f(t(\x_1), \ldots, t(\x_m))]_{\cong}, \] so that $f(\x_1, \ldots, \x_m) \cong f(t(\x_1), \ldots, t(\x_m))$ and hence \[ [f(\x_1, \ldots, \x_m)]_{\cong}^{k'} \approx_{k'} [f(t(\x_1), \ldots, t(\x_m))]_{\cong}^{k'}. \] Now, the alien subterms of $f(\x_1, \ldots, \x_m)$ are $\x_1, \ldots, \x_m$, and those of $f(t(\x_1), \ldots, t(\x_m))$ are $t(\x_1), \ldots, t(\x_m)$. For any $1 \leq i \neq j \leq m$ we have $\x_i \not\cong \x_j$, since $[\x_i]_{\cong}^4 = \x_i \not\approx_4 \x_j = [\x_j]_{\cong}^4$. Also, by Lemma \ref{distinctindeterminateslemma} we have $t(\x_i) \not\cong t(\x_j)$ for $1 \leq i \neq j \leq m$. Finally, for any $1 \leq i \leq m$ there is no $1 \leq j \leq m$ such that $t(\x_i) \cong \x_j$, because $\mroot(t(\x_i)) \in \Sigma_k$ and $\mroot(\x_j) \in \Sigma_4$ and $k \neq 4$. So the terms $\x_1, \ldots, \x_m, t(\x_1), \ldots, t(\x_m)$ are pairwise non-isomorphic. Now let $y_1, \ldots, y_m, z_1, \ldots, z_m$ be pairwise distinct variables. Then because \[ [f(\x_1, \ldots, \x_m)]_{\cong}^{k'} \approx_{k'} [f(t(\x_1), \ldots, t(\x_m))]_{\cong}^{k'}, \] Lemma \ref{complexlemma} implies that
\[ \T_{k'} \vdash f(y_1, \ldots, y_m) = f(z_1, \ldots, z_m). \] But this means that $f$ is constant in $\T_{k'}$, contrary to assumption. So statement (3) cannot be true for $t$. 

Statement (4) cannot be true for $t$, because if we had $t[f(\overline{\x})/\x]^{\mathcal{A}} = [\x_i]_{\cong}$ for some $i \geq 1$, then we would have $[\x_i]_{\cong} = t[f(\overline{\x})/\x]^{\mathcal{A}} = [f(t(\x_1), \ldots, t(\x_m))]_{\cong}$, and hence $f(t(\x_1), \ldots, t(\x_m)) \cong \x_i$, which is impossible because $f$ and $\x_i$ belong to different signatures. 

Lastly, suppose that statement (5) were true for $t$. So \[ t[f(\overline{\x})/\x]^{\mathcal{A}} = [D(t_1', \ldots, t_r')]_{\cong}, \] where $D$ is a proper $\Sigma_{k'}$-context for $k' \neq k \in \{1, 2\}$ and for every $1 \leq q \leq r$, there is a term $t_q$ of rank $\leq \rank(t) - 2$ such that $\mroot(t_q) \notin \Sigma_{k'}$ and $t_q[f(\overline{\x})/\x]^{\mathcal{A}} = [t_q']_{\cong}$. By Lemma \ref{existencesmallerranklemma}, we may assume without loss of generality that $\rank(t_q') \leq \rank(t_q[f(\overline{\x})/\x])$ for each $1 \leq q \leq r$. So for every $1 \leq q \leq r$ we have
\begin{align*}
\rank(t_q')	&\leq \rank(t_q[f(\overline{\x})/\x]) \\
		&\leq \rank(t_q) + 1 \\
		&\leq (\rank(t) - 2) + 1 \\
		&= \rank(t) - 1,
\end{align*}
\noindent where the second inequality comes from Lemma \ref{functionsymbollemma}. Now, there \emph{may} be some $1 \leq q \leq r$ such that $t_q'$ is \emph{not} an alien subterm of $D(t_1', \ldots, t_r')$, so let $D(t_1', \ldots, t_r') \equiv D'[w_1, \ldots, w_s]$ for some proper $\Sigma_{k'}$-context $D'$ and closed terms $w_1, \ldots, w_s$ such that for any $1 \leq p \leq s$, there is some $1 \leq q \leq r$ such that either $w_p \equiv t_q'$ (if $\mroot(t_q') \notin \Sigma_{k'}$) or $w_p$ is an alien subterm of $t_q'$ (if $\mroot(t_q') \in \Sigma_{k'}$). So for any $1 \leq p \leq s$, there is some $1 \leq q \leq r$ such that $\rank(w_p) \leq \rank(t_q') \leq \rank(t) - 1$. 

Now we have \[ [D'[w_1, \ldots, w_s]]_{\cong} = [D(t_1', \ldots, t_r')]_{\cong} = t[f(\overline{\x})/\x]^{\mathcal{A}} = [f(t(\x_1), \ldots, t(\x_m))]_{\cong}, \] and hence \[ D'[w_1, \ldots, w_s] \cong f(t(\x_1), \ldots, t(\x_m)), \] so that

\[ [D'[w_1, \ldots, w_s]]_{\cong}^{k'} \approx_{k'} [f(t(\x_1), \ldots, t(\x_m))]_{\cong}^{k'}. \] Now we want to show that for every $1 \leq i \leq m$, there is no $1 \leq p \leq s$ with $t(\x_i) \cong w_p$. To prove this, we will require the following definitions and lemma.

\begin{definition}
For any closed term $u$ over $\Sigma$, we define the set $\Ind(u)$ of all indeterminates occurring in $u$, by recursion on the structure of $u$:
\begin{itemize}

\item If $u$ is a constant symbol of $\Sigma$, then either $u \in \Sigma_4$ or $u \notin \Sigma_4$. In the first case, either $u \equiv \x$, in which case $\Ind(u) := \{\x\}$, or $u \equiv \x_i$ for some $i \geq 1$, in which case $\Ind(u) = \{\x_i\}$. Otherwise, if $u \notin \Sigma_4$, we set $\Ind(u) := \emptyset$.

\item If $u \equiv g(u_1, \ldots, u_m)$ for some function symbol $g \in \Sigma$ of arity $m \geq 1$ and closed terms $u_1, \ldots, u_m$ over $\Sigma$, then \[ \Ind(u) = \Ind(g(u_1, \ldots, u_m)) := \Ind(u_1) \cup \ldots \cup \Ind(u_m). \] \qed
\end{itemize}
\end{definition}

\begin{definition}
For any closed term $u$ over $\Sigma$, we define $u[\x/\Ind(u)]$, i.e. the term $u$ with $\x$ replacing all indeterminates of $u$, by recursion on the structure of $u$ as follows:
\begin{itemize}

\item If $u$ is a constant symbol of $\Sigma$, then either $u \in \Sigma_4$ or $u \notin \Sigma_4$. If $u \in \Sigma_4$, then $u[\x/\Ind(u)] := \x$. If $u \notin \Sigma_4$, then $u[\x/\Ind(u)] := u$.

\item If $u \equiv g(u_1, \ldots, u_m)$ for some function symbol $g \in \Sigma$ of arity $m \geq 1$ and closed terms $u_1, \ldots, u_m$ over $\Sigma$, then 
\[ u[\x/\Ind(u)] := g(u_1[\x/\Ind(u_1)], \ldots, u_m[\x/\Ind(u_m)]). \] \qed
\end{itemize}
\end{definition}

\noindent One can easily verify that if $u$ does not contain any indeterminates (i.e. if $\Ind(u) = \emptyset$), then $u[\x/\Ind(u)] \equiv u$, and that if $u \equiv C[s_1, \ldots, s_\ell]$ for some proper $\Sigma_k$-context $C \in \Term(\Sigma_k, V)$ ($k \in \{1, 2\}$) and closed terms $s_1, \ldots, s_\ell \in \Term^c(\Sigma)$ with $\mroot(s_j) \notin \Sigma_k$ for each $1 \leq j \leq \ell$, then \[ u[\x/\Ind(u)] \equiv C[s_1[\x/\Ind(s_1)], \ldots, s_\ell[\x/\Ind(s_\ell)]]. \] The proof of the following lemma is then almost identical to the proof of Lemma \ref{indeterminatelemma}: 

\begin{lemma}
\label{replacingindeterminateslemma}
{\em
For any closed terms $u, v$ over $\Sigma$, \[ u \cong v \ \Longrightarrow \ u[\x/\Ind(u)] \cong v[\x/\Ind(v)]. \] \qed
}
\end{lemma}

We can now show that for every $1 \leq i \leq m$, there is no $1 \leq p \leq s$ with $t(\x_i) \cong w_p$. Suppose towards a contradiction that $t(\x_i) \cong w_p$ for some $1 \leq i \leq m$ and $1 \leq p \leq s$. Then by Lemma \ref{replacingindeterminateslemma}, we obtain \[ t(\x_i)[\x/\Ind(t(\x_i))] \cong w_p[\x/\Ind(w_p)]. \] Because $\Ind(t) = \{\x\}$, it follows that $\Ind(t(\x_i)) = \{\x_i\}$, so that $t(\x_i)[\x/\Ind(t(\x_i))] \equiv t(\x_i)[\x/\x_i] \equiv t$. So then we have $t \cong w_p[\x/\Ind(w_p)]$. Now, we know that $w_p[\x/\Ind(w_p)]$ is a term over $\Sigma_1 \cup \Sigma_2 \cup \Sigma_3 \cup \{\x\}$, and by Lemma \ref{rankindeterminateslemma} below, we have $\rank(w_p[\x/\Ind(w_p)]) = \rank(w_p) < \rank(t)$. So $t$ is isomorphic to a term over $\Sigma_1 \cup \Sigma_2 \cup \Sigma_3 \cup \{\x\}$ of strictly smaller rank, which contradicts our assumptions on $t$. It follows that there can be no $1 \leq i \leq m$ and $1 \leq p \leq s$ with $t(\x_i) \cong w_p$, as desired. 

\begin{lemma}
\label{rankindeterminateslemma}
{\em
For any closed term $u$ over $\Sigma$, \[ \rank(u[\x/\Ind(u)]) = \rank(u). \]
}
\end{lemma}

\begin{proof}
We prove this by induction on $\rank(u)$. First suppose that $\rank(u) = 0$, so that $u$ is pure. So either $u$ is an indeterminate, or $u$ does not contain any indeterminates. If $u$ is an indeterminate, then we have 
\[ \rank(u[\x/\Ind(u)]) = \rank(\x) = 0 = \rank(u). \] Otherwise, we have $u[\x/\Ind(u)] \equiv u$, from which the result obviously follows. 

Now suppose that $\rank(u) > 0$, and let $u \equiv C[s_1, \ldots, s_\ell]$ for some proper $\Sigma_k$-context $C$ (for some $k \in \{1, 2\}$) and closed terms $s_1, \ldots, s_\ell$ over $\Sigma$ with $\mroot(s_j) \notin \Sigma_k$ for each $1 \leq j \leq \ell$. Since $\rank(s_j) < \rank(u)$ for each $1 \leq j \leq \ell$, the induction hypothesis yields $\rank(s_j[\x/\Ind(s_j)]) = \rank(s_j)$ for each $1 \leq j \leq \ell$. Then we have
\begin{align*}
\rank(u[\x/\Ind(u)])	&= \rank(C[s_1[\x/\Ind(s_1)], \ldots, s_\ell[\x/\Ind(s_\ell)]]) \\
				&= 1 + \bmax\{\rank(s_1[\x/\Ind(s_1)]), \ldots, \rank(s_\ell[\x/\Ind(s_\ell)])\} \\
				&= 1 + \bmax\{\rank(s_1), \ldots, \rank(s_\ell)\} \\
				&= \rank(C[s_1, \ldots, s_\ell]) \\
				&= \rank(u),
\end{align*}
\noindent as required. This completes the induction and the proof.
\end{proof}

\noindent Recall that we have \[ [D'[w_1, \ldots, w_s]]_{\cong}^{k'} \approx_{k'} [f(t(\x_1), \ldots, t(\x_m))]_{\cong}^{k'}. \] By Lemma \ref{distinctindeterminateslemma} and the conclusion reached just before Lemma \ref{rankindeterminateslemma}, Lemma \ref{complexlemma} implies that
\[ \T_{k'} \vdash f(y_1, \ldots, y_m) = v, \]
\noindent where $y_1, \ldots, y_m$ are distinct variables and $v \in \Term(\Sigma_{k'}, V)$ does not contain any of $y_1, \ldots, y_m$. But this means that $f$ is constant in $\T_{k'}$, contrary to assumption. \par

This completes the proof that none of the statements (1) -- (5) of Lemma \ref{mostimportantlemma} can be true for $t$. Since this contradicts Lemma \ref{mostimportantlemma}, this means that our assumption that $t \not\equiv \x$ was wrong, and hence we must have $t \equiv \x$ and \emph{a fortiori} $t \sim_{M_n, \x} \x$ after all. So $[t] \in G_{\T_1 + \T_2}(M_n)$ implies $[t] = [\x]$, which implies that $G_{\T_1 + \T_2}(M_n) = [\x]$, as desired.
\end{proof}

\noindent We can now deduce the following characterization of the \emph{categorical} isotropy groups of the free, finitely generated models of $\T_1 + \T_2$, which follows immediately from Theorem \ref{maintheorem} and the two bullet points at the end of Section 1. 

\begin{corollary}
{\em Let $\T_1, \T_2$ be equational theories over respective disjoint signatures $\Sigma_1, \Sigma_2$, and suppose for each $i \in \{1, 2\}$ that $\T_i$ has at least one function symbol that is neither constant nor a projection in $\T_i$. Let $n \geq 0$, with $M_n$ the free model of $\T_1 + \T_2$ on $n$ generators. 

\begin{itemize}
\item Let \[ \pi = \left(\pi_h : \mathsf{cod}(h) \to \mathsf{cod}(h)\right)_{\mathsf{dom}(h) = M_n} \] be a (not necessarily natural) family of endomorphisms in $(\T_1 + \T_2)\mathsf{mod}$ indexed by morphisms with domain $M_n$. Then $\pi \in \mathcal{Z}_{\T_1 + \T_2}(M_n)$ iff 
\[ \pi_h = \mathsf{id}_M : M \xrightarrow{\sim} M \] for each $h : M_n \to M$ in $(\T_1 + \T_2)\mathsf{mod}$.

\item Let $h : M_n \xrightarrow{\sim} M_n$ be an automorphism of $M_n$. Then $h$ is a (categorical) inner automorphism of $M_n$ iff $h = \mathsf{id}_{M_n}$.  
\end{itemize} \qed
} 
\end{corollary}

\noindent Therefore, the only `coherently extendible' automorphism of the free model of $\T_1 + \T_2$ on $n$ generators is the \emph{identity} automorphism.

For a single-sorted equational theory $\T$, let $\Aut\left(\Id_{\Tmod}\right)$ be the group of all natural automorphisms of the identity functor $\Id_{\Tmod} : \Tmod \to \Tmod$ of $\Tmod$ (which is also the group of invertible elements of the \emph{centre} of the category $\Tmod$, which is the monoid $\mathsf{End}\left(\Id_{\Tmod}\right)$ of natural \emph{endo}morphisms of $\Id_{\Tmod}$). We may refer to $\Aut\left(\Id_{\Tmod}\right)$ as the \emph{global isotropy group} of $\T$. Since $M_0$, the free $\T$-model on $0$ generators, i.e. the \emph{initial} $\T$-model, is an initial object of $\Tmod$, it is not difficult to see that
\[ \Aut\left(\Id_{\Tmod}\right) = \mathcal{Z}_\T(M_0). \] We thus obtain:

\begin{corollary}
{\em Let $\T_1, \T_2$ be equational theories over respective disjoint signatures $\Sigma_1, \Sigma_2$, and suppose for each $i \in \{1, 2\}$ that $\T_i$ has at least one function symbol that is neither constant nor a projection in $\T_i$. Then the global isotropy group
\[ \Aut\left(\Id_{(\T_1+\T_2)\mathsf{mod}}\right) = \mathcal{Z}_{\T_1 + \T_2}(M_0) \]
of $(\T_1 + \T_2)\mathsf{mod}$ is trivial. 

So the only natural automorphism of $\Id_{(\T_1+\T_2)\mathsf{mod}}$ is the identity natural transformation. 
\qed
}
\end{corollary}   

\begin{remark}
The assumption in Theorem \ref{maintheorem} that both $\T_1$ and $\T_2$ have function symbols that are neither constant nor projections is necessary to obtain the result in general. For a trivial counterexample to Theorem \ref{maintheorem} when this assumption is \emph{not} satisfied, let $\T_1$ be any equational theory for which there is some $n \geq 0$ such that the free model of $\T_1$ on $n$ generators has non-trivial isotropy group (e.g. the theory of groups, cf. \cite[Example 4.1]{MFPSpaper}, \cite[Corollary 3.2.9]{thesis}), and let $\T_2$ be the empty theory over the empty signature. Then $\T_2$ clearly does not have any function symbol that is neither constant nor a projection, and moreover $\T_1 + \T_2 = \T_1$, so the free model of $\T_1 + \T_2 = \T_1$ on $n$ generators will have non-trivial isotropy group, contrary to the result of Theorem \ref{maintheorem}.

For a slightly less trivial counterexample (where neither theory is empty), if $\T_1$ is a theory with the same properties as in the first counterexample, and $\T_2$ is an equational theory over a non-empty signature $\Sigma_2$ containing only function symbols of arity $\geq 1$ that are all \emph{projections} in $\T_2$, then it is easy to see that if $t \in \Term^c(\Sigma_1(\x, \y_1, \ldots, \y_n))$ and $[t] \in G_{\T_1}(M_{\T_1, n})$ (where $M_{\T_1, n}$ is the free model of $\T_1$ on $n$ generators) and $[t] \neq [\x]$, then $[t]$ will commute generically with every function symbol of $\Sigma_2$ (because these are all projections in $\T_2$), so that we will have $[t] \in G_{\T_1 + \T_2}(M_n)$ and $[t] \neq [\x]$, contrary to the result of Theorem \ref{maintheorem}. \qed     
\end{remark}

\section{Conclusions}

Using techniques and results from \cite[Chapter 9]{Rewriting} and \cite{Combining}, which essentially show that the word problem for free models of a disjoint union of equational theories is solvable if the word problems for the free models of the component theories are solvable, we have shown that the isotropy groups of all free, finitely generated models of a disjoint union of equational theories (satisfying minimal assumptions) are \emph{trivial}. A natural next step would be to try to extend the results herein to apply to all \emph{finitely presented} models of $\T_1 + \T_2$; however, given the apparent lack of results like those in \cite{Combining} for \emph{finitely presented} models of unions of equational theories, there is no obvious way of doing this.

\end{document}